\documentclass[11pt]{article} 
 
\usepackage{latexsym} 
\usepackage{amsfonts} 
\usepackage{amsmath} 
\usepackage{amsthm} 
\usepackage{mathrsfs} 
\usepackage{dsfont} 
\usepackage{bbold} 
\usepackage[english]{babel} 
\usepackage{caption} 
\usepackage{epsfig} 
\usepackage{float} 
\usepackage{subfigure} 
\usepackage{psfrag} 
\usepackage{graphicx} 
\usepackage{epsfig} 
\usepackage{amsbsy} 
\usepackage{comment}
\usepackage{xcolor}
\usepackage{marginnote}
\usepackage[hidelinks,pdfencoding=auto]{hyperref}
\usepackage[normalem]{ulem}

\numberwithin{equation}{section}

\def\red{\textcolor{red}}

\def\red{}

\DeclareMathOperator{\sgn}{sgn}
\DeclareMathOperator{\Res}{Res}

\newtheorem{theorem}{Theorem} 
\newtheorem*{prop*}{Theorem} 
\newtheorem{theo}{Theorem}
\newtheorem{coro}[theorem]{Corollary} 
 
\newtheorem{lemma}[theorem]{Lemma} 
 
\newtheorem{rmk}[theorem]{Remark}

 \numberwithin{theorem}{section}

\newcommand{\beq}{\begin{equation}}
\newcommand{\eeq}{\end{equation}}
\newcommand{\matht}[1]{\texorpdfstring{#1}{}}

\newcommand{\matC}{{\mathscr C}}

\newcommand{\matG}{{\mathscr G}}

\newcommand{\matS}{{\mathscr S}}

\newcommand{\e}{\varepsilon}

\newcommand{\dam}{\beta}

\newcommand{\x}{\xi}

\def\tilde#1{\widetilde{#1}}
 
\def\ins#1#2#3{\vbox to0pt{\kern-#2 \hbox{\kern#1 #3}\vss}\nointerlineskip}

\def\sp{g}
\def\Sp{G}
 
\begin{document}
 
\title{\bf Approach to equilibrium for a particle interacting with
  a harmonic thermal bath}

\author{{\bf Federico Bonetto$^{1,*}$, Alberto Maiocchi$^2$}\\
	\small $^1$ School of Mathematics, Georgia Institute of Technology,Atlanta, GA 30332, USA\\
	\small $^2$ Dipartimento di Matematica e Applicazioni, Universit\`a di Milano -- Bicocca, Milano, Italy\\
	\small $^*$ corresponding author: \tt fb49@gatech.edu
}
\date{} 
 
\maketitle 
 
\begin{abstract} 
  We study the long time evolution of the position-position correlation function $C_{\alpha,N}(s,t)$ for a harmonic
oscillator (the {\it probe}) interacting via a coupling $\alpha$ with a large chain of $N$ coupled oscillators (the
{\it heat bath}). At $t=0$ the probe and the bath are in equilibrium at temperature $T_P$ and $T_B$, respectively. We
show that for times $t$ and $s$ of the order of $N$, $C_{\alpha,N}(s,t)$ is very well approximated by its limit
$C_{\alpha}(s,t)$ as $N\to\infty$. We find that, if the frequency $\Omega$ of the probe is in the spectrum of the
bath, the system appears to thermalize, at least at higher order in $\alpha$. This means that, at order 0 in $\alpha$,
$C_\alpha(s,t)$ equals the correlation of a probe in contact with an ideal stochastic {\it thermostat}, that is forced
by a white noise and subject to dissipation. In particular we find that $\lim_{t\to\infty}
C_\alpha(t,t)=T_B/\Omega^2$ while that $\lim_{\tau\to\infty} C_\alpha(\tau,\tau+t)$ exists and decays exponentially in
$t$. Notwithstanding this, at higher order in $\alpha$, $C_{\alpha}(s,t)$ contains terms that oscillate or vanish as a
power law in $|t-s|$. That is, even when the bath is very large, it cannot be thought of as a stochastic thermostat. 
When the frequency of the bath is far from the spectrum of the bath, no thermalization is observed.
\end{abstract} 
  
 \tableofcontents
 
\section{Introduction}
\label{sec:intro}

Several authors studied the approach to equilibrium or the non equilibrium steady state of systems in contact with one
or more thermostats, see for example \cite{RLL67} and \cite{BLV}. The thermostats are normally modeled as the
idealization of the interaction with a large, potentially infinite, heat reservoir. Instead of the large number of
degrees of freedom needed to describe the reservoir, the thermostat can be modeled via an effective interaction with a
low dimensional stochastic (or sometimes deterministic, see e.g. \cite{EH85,CELS,BDLR}) process evolving independently.
These idealizations have proved very useful in studying properties of non--equilibrium statistical mechanics.

In more recent times, some authors have tried to derive these idealized low dimensional thermostats from the evolution
of large, possibly infinite, heat reservoirs fully coupled with the system of interest. Examples of this kind can be
found in \cite{EPR-B99} for an out of equilibrium anharmonic chain or in \cite{BLTV} for the approach to equilibrium of
a simple kinetic model. Contrary to \cite{EPR-B99}, in \cite{BLTV} the reservoir is represented as a large but finite
gas initially found in canonical equilibrium at temperature $T$. This is compared with a so called Maxwellian
thermostat, i.e. an idealized infinite gas reservoir, see \cite{BLV}. Since equilibration times are normally much longer
than the natural timescale of the microscopic dynamics, it is important to carefully control the difference between the
finite reservoir and the idealized infinite one for long times or, ideally, uniformly in time.

Another simple example of heat bath one can think of is formed by a large array of coupled harmonic oscillators \red{of
equal mass $m$}, like for example a chain or a higher dimensional finite lattice. One of the oscillators in such an
array \red{is characterized by a different mass $M>m$ and} can be considered as the system while all the others act as a
thermal bath. \red{This system is sometime referred to as Rubin model \cite{Rubin1,Rubin2}. The long time behavior of
the  correlation functions in the situation where both the system and the thermal bath are initially in equilibrium
(either at equal or different temperatures) was extensively studied analytically in \cite{Rubin1,Rubin2,Caldeira} and
from a more numerical point of view in \cite{Das,Dhar,Onofrio}. Further analysis can be found in \cite{TK62,TH62} and in
particular in \cite{FKM65} where the main interest is whether, or in which condition, the effective dynamics of the
oscillator of mass $M$,} once the size of the heat reservoir is sent to infinity, can be described by a stochastic
differential equation. Indeed, one expects that when the reservoir is large enough it can be effectively modeled by
the interaction with a white noise together with a dissipation term to prevent the system from overheating. From a
physical point of view this means that one can neglect the effect of the interaction of the small system on the large
reservoir that thus evolves autonomously. In such a situation the system sees the reservoir as a stochastic force while
the counter action on the reservoir appears as a dissipative term.

A somehow related point of view is discussed in \cite{EPR-B99} where the heat reservoirs are modeled as two infinite
scalar fields interacting with the first and last oscillator in a chain through a dipole style linear term. Again
the fields are assumed to be initially in canonical equilibrium at different temperature. By formally integrating the
equation of motion of the fields the authors obtain a set of stochastic differential equations that describes a colored
noise thermostat.

\red{We consider here what is possibly the simplest model for an out of equilibrium system formed by a large but finite
{\it thermal bath} and an external {\it probe}, initially separately in equilibrium at different temperatures. The
thermal bath is modeled as a chain of $N$ equal particles of mass $m$ linked to their nearest neighbor by springs of
strength $\sp$ and pinned to their equilibrium position by springs of strength $\sp'$, \cite{Rubin1}. }The probe is
modeled as a particle of mass $M$ pinned to its equilibrium position by a spring of strength $\Sp$. The bath is
initially in equilibrium at a temperature $T_B$ while the probe is in equilibrium at temperature $T_P$. At time $t=0$
the probe is put in contact with the bath by connecting it via a spring of strength $\alpha$ to one of the oscillators
of the bath \red{so that, for $t>0$, the states, and thus the temperatures, of the bath and the probe will change. See 
Subsection \ref{subsec:model} for a precise definition of our model.}

We are interested in the long time evolution of the probe in the limit when $N$ grows to infinity with $\alpha$
independent from $N$. \red{Moreover, we want to obtain quantitative estimates of the correction to the limiting behavior
when $N$ is large but finite}. Two regimes naturally appears. In the {\it non-resonant} regime, when the natural
frequency $\Omega$ of the probe is found outside the frequency spectrum $[\mu_-,\mu_+]$ of the normal modes of the bath,
the interaction is ineffective and the evolution of the probe is a small modification of its {\it unperturbed} (that is
$\alpha=0$) evolution. In particular, no thermalization takes place. On the other hand, in the {\it resonant} regime,
when $\Omega$ is found inside $[\mu_-,\mu_+]$ and $N$ is very large, it is natural to assume that the probe will {\it
thermalize}, that is it will equilibrate at temperature $T_B$ while the state of the bath will barely change. In both
cases, on a more detailed level, one expects that, still for large $N$, the bath can be seen as an external noise acting
on the probe so that the effective evolution of the probe can be described as a Markov process in which the bath has
been replaced by an effective low dimensional stationary stochastic process, that is a {\it thermostat}.

We study these questions by looking at the two times position-position correlation functions $C_{\alpha,N}(s,t)$ for the
probe. \red{As in \cite{Rubin1,Rubin2} or more recently \cite{Plyu}, the main tool to study the asymptotic behavior of $C_{\alpha,N}(s,t)$ 
is its
Laplace transform $\widetilde C_{\alpha,N}(\lambda,\lambda')$.} In Section \ref{sec:solution}, we compute the Laplace
transform of the solution of the equation of motion for the probe. We use this solution in Sections \ref{sec:decay1} and
\ref{sec:decay2} to obtain an explicit expressions for $\widetilde C_{\alpha,N}(\lambda,\lambda')$ and  then take the
limit as $N\to\infty$ of these expressions obtaining the Laplace transform $\widetilde C_{\alpha}(\lambda,\lambda')$ of
the effective correlation functions for the infinite system. Since $C_{\alpha,N}(s,t)$ is a quasi-periodic function, in
general one cannot expect it to be close to its limit uniformly in time. The inverse of the spacing of the frequencies
of the normal modes of the bath provides a natural time scale for a comparison between the finite and infinite system.
We can thus show that, for times short when compared with $N/(\mu_+-\mu_-)$, the inverse Laplace transform
$C_{\alpha}(s,t)$ of $\widetilde C_{\alpha}(\lambda,\lambda')$ approximates extremely well $C_{\alpha,N}(s,t)$ with
correction exponentially small in $N$, see Section \ref{sec:f} and Subsections \ref{subsec:bound} and \ref{subsec:C22}.
Finally, Sections in \ref{sec:decay1} and \ref{sec:decay2}, we use the explicit expression for $\widetilde
C_{\alpha}(\lambda,\lambda')$ to obtain detailed information on the  asymptotic behavior of $C_{\alpha}(s,t)$. \red{This
behavior clearly depends on the value of $\alpha$. For $\alpha$ small it can be written as a main part, largely
independent of $\alpha$, plus correction of higher order that we compute with some detail, see in particular Appendix
\ref{app:stime}.} From the analysis of $C_{\alpha,N}(s,t)$ it is then easy to obtain analogous information for the
momentum-momentum correlation function $D_{\alpha,N}(s,t)$.

About the two expectations discussed above we find that the first one is met. This means that, in the resonant case, the
probe appears to thermalize to the temperature of the bath in the sense that the average kinetic energy
$D_{\alpha}(t,t)$ and average internal energy $D_{\alpha}(t,t)+\Omega^2 C_{\alpha}(t,t)$ of the probe converge
exponentially fast in $t$ to values close (but not equal) to those predicted by an equilibrium state at temperature
$T_B$. On the other hand, a Markovian evolution toward a steady state would imply that $\lim_{\tau\to\infty}
C_\alpha(\tau,\tau+t)$ exists and decays exponentially in $t$.  Although this is true if one only looks at the term of
order zero in $\alpha$, strictly speaking neither of these implications is true since $C_\alpha(\tau,\tau+t)$ contains
terms oscillating in $\tau$ and terms decaying as a power law in $\tau$ and $\tau+t$. \red{Our results are thus
consistent with the numerical results in \cite{Onofrio}. The implications for possible numerical simulations of the
detailed behavior in $N$ $\alpha$ and $t$ of the correlation functions is briefly discussed in Section
\ref{sec:discussion}.}

\red{To better compare with previous works, see \cite{Mori65, Weiss}, we show that the evolution of the position of the
probe can be described using a generalized Langevin equation characterized by a delayed ``dissipation like" term and a
random forcing. The delayed term and the forcing satisfy a relation analogous to the classical fluctuation
dissipation theorem.} We get some further insight by studying the stochastic system obtained by replacing the delayed 
dissipation with a constant friction, that is by neglecting the backward interaction of the probe on
the bath. This analysis shows that the presence of terms decaying as a power law are due to the finite frequency 
spectrum of the bath 
while the presence of oscillatory corrections to the asymptotic behavior of
$C_\alpha(s,t)$ is due to the backward interaction of the probe on the bath. Notwithstanding the fact that these
corrections are of higher order in $\alpha$, they persist even when $N\to \infty$.

The rest of the paper is organized as follows. In Section \ref{sec:model} we introduce the exact model we will study and
our main results while in Section \ref{sec:discussion} we outline some of the possible extensions and open problems of
our work. Sections \ref{sec:solution} to \ref{sec:decay2} contain the proofs of our results while Appendix
\ref{app:lemmas} contains the statement and proof of several Lemmas useful throughout the paper. Finally Appendix
\ref{app:stime} contains the technical details for the improved estimate of the correction terms to $C_{\alpha}(s,t)$.
 
\section{Setting and main results} 
\label{sec:model}

In this section we first introduce the model we will study in the rest of the paper. We then discuss our main results
and finally we compare them with the analogous results for a system where the large thermal bath is replaced by a
suitable stochastic thermostat.

\subsection{The model}\label{subsec:model}

We consider a system of $2N+1$ linear oscillators, where all but one oscillator (the {\it bath}) act as a heat bath on
the remaining one (the {\it probe}). We model the bath as a chain of identical particles with nearest neighbor
interaction and on--site pinning potential, with one of the particles is linked with the probe. The Hamiltonian of
the system is thus
\begin{equation}\label{eq:Ham0}
\begin{aligned}
H(\hat q,\hat p,\bar Q,\bar P):=& H_B(\hat q,\hat p)+H_P(\bar Q,\bar P)+\alpha H_I(\hat q,\hat p,\bar Q,\bar P):=\\
&\sum_{l=-N+1}^{N} \left(\frac{\hat
  p_l^2}{2m}+\frac{\sp}2(\hat q_{l+1}-\hat
q_l)^2+\sp'\frac{\hat q_l^2}2\right) + \frac{\bar
  P^2}{2M}+\Sp\frac{\bar Q^2}2
+ \frac{\alpha'}2(\bar Q-\hat q_0)^2 ,
\end{aligned}
\end{equation}
where the Hamiltonian $H_B$ and the canonically conjugated variables $\bar Q$, $\bar P$ pertain to the probe, while
$H_P$ and the canonically conjugated variables $\{\hat q_l\}_{l=-N+1,\ldots,N}$, $\{\hat p_l\}_{l=-N+1,\ldots,N}$ are
the Hamiltonian and the coordinates of the $2N$ particles of the bath, with periodic boundary condition, that is $\hat
q_N=\hat q_{-N}$. Finally $H_I$ describes the interactions between bath and probe. The elastic constant of the probe is
denoted by $\Sp$, while the coupling between the probe and the bath is tuned by the parameter $\alpha'$. We think of the
connection between the probe and the particle in the chain as a spring so that it is natural to require $\alpha'>0$. We
note though that most of our results remain true for $\alpha'<0$ as far as the origin remain a stable fixed point for
\eqref{eq:Ham0}. 

As a first step, we pass to the normal modes of oscillation for the bath, that is, we define the canonically conjugated
variables $\{q_j\}_{j=-N+1,\ldots,N}$, $\{p_j\}_{j=-N+1,\ldots,N}$ through $q_j=\sqrt m\sum_lO_{jl}\hat q_l$ and
$p_j=\left(1/\sqrt m\right)\sum_lO_{jl}\hat p_l$, where the orthogonal matrix $O$ is defined by %
\begin{equation}\label{eq:normal}
O_{lj}= \left\{\begin{array}{cc}
\frac{\eta_l}{\sqrt N}\cos \left(\frac{jl\pi}{N}\right)& l=0,\ldots,N\\
\frac{1}{\sqrt N}\sin \left(\frac{jl\pi}{N}\right)& l=-N+1,\ldots,-1
\end{array}
\right.\ , \quad \mbox{with }\eta_l=\left\{\begin{array}{cc}\frac1{\sqrt2}&
l=0,N\\
1& \mbox{elsewhere}\end{array}\right.\ .
\end{equation}
By inversion, we get
\[
\hat q_0=\frac1{\sqrt{Nm}} \sum_{j=0}^N\eta_j q_j,
\]
so that the dynamics of the odd normal modes of the baths, that is the normal modes indexed by $j$ ranging from $-N+1$
to $-1$ in \eqref{eq:normal}, is decoupled from the rest of the system. Thus, from now on, we restrict our attention to
the system composed by the probe and the even normal modes of the bath. This is equivalent to a system of $N+1$
oscillators (indexed with $j$ ranging from 0 to $N$) plus the probe and is described by the Hamiltonian

\begin{align}\label{eq:Ham1}
H(q,p,Q,P)=&\sum_{j=0}^{N} \frac{p_j^2}2 + \sum_{j=0}^{N} \frac{\omega_j^2 
q_j^2}2 + \frac{P^2}2+\frac{\Omega^2Q^2}2 +
\frac\alpha{2}\left(\sqrt{\frac{\gamma}{N}}\sum_{j=0}^N \eta_j q_j-
\frac1{\sqrt\gamma}Q\right)^2=\\
&\red{\sum_{j=0}^{N} \frac{p_j^2}2 + \sum_{j=0}^{N} \frac{\omega_j^2 
	q_j^2}2 + \frac{\alpha\gamma}{2N}\biggl(\sum_j \eta_j q_j\biggr)^2+ \frac{P^2}2+\frac{\overline \Omega^2Q^2}2 +
\frac{\alpha Q}{\sqrt N}\sum_{j=0}^N \eta_j q_j}\label{eq:Ham1.1}
\end{align}
where we have introduced the rescaled canonical variables $P:=\bar P/\sqrt M$ and $Q:=\sqrt M \bar Q$ for the probe and 
the frequencies $\omega_j$ of the bath are given by
\begin{equation}\label{eq:frequenze_catena}
\omega_j :=
\sqrt{\mu_-^2+4\tilde\omega^2\sin^2\left(\frac{j\pi}{2N}\right)}\ ,
\end{equation}
with $\tilde\omega^2:=\sp/m$ and $\mu_-^2:=\sp'/m$, and $\alpha:=\alpha'/\sqrt {mM}$, $\Omega^2:=\Sp/M$, $\gamma :=
\sqrt{M/m}$. \red{We have also introduced the {\it dressed probe frequency} 
$\overline \Omega^2=\Omega^2 +\alpha/\gamma$ that include the corrections order $\alpha$ to $\Omega$, 
see 
\eqref{eq:polone}.}

Observe that $\omega_0=\mu_-$ while
\[
\mu_+:=\omega_N=\sqrt{\mu_-^2+4\tilde\omega^2}.
\]
Moreover we will write $\omega_j=\omega(\theta_j)$ with 
\begin{equation}\label{eq:omega-theta}
\theta_j:=\frac{j\pi}{N}\qquad\hbox{and}\qquad\omega(\theta):=\sqrt{\mu_-^2+2\tilde\omega^2(1-\cos(\theta))}\, .
\end{equation}
Initially the state of the system is represented by the product of a Maxwellian distribution at temperature $T_B$ for
the bath times a Maxwellian distribution at temperature $T_P$ for the probe. Since the change of variable $O$ in
\eqref{eq:normal} is orthogonal, the initial density can be written as
\begin{equation}\label{eq:dens_prob}
\rho_{N}(q,p,Q,P)=\frac{1}{Z(T_B,T_P)} \exp\left(-\frac 1{2T_B}\sum_{j=0}^N
(p_j^2+\omega_j^2 q_j^2)-\frac 1{2T_P} (P^2+\Omega^2Q^2)\right)\ ,
\end{equation}
where $Z(T_B,T_P)$ is the partition function.

\subsection{Main results}

Let $Q(t)$ and $P(t)$ be the position and momentum of the probe when the system starts with initial condition
$q(0),p(0),Q(0)$ and $P(0)$. We will focus our attention on the 2-times correlation functions for the probe. In
particular we will study the position-position correlation function defined as
\begin{equation}\label{eq:corr}
 C_{\alpha, N}(s,t):=\langle Q(s) Q(t)\rangle_{N},
\end{equation}
where $\langle\cdot\rangle_{N}$ represents the average over the initial condition with respect to the probability
density $\rho_{N}$. From this we will obtain the momentum-momentum correlation function $D_{\alpha, N}(s,t)$, the average 
kinetic energy $E_{\alpha, N}(t)$ and average energy $U_{\alpha, N}(t)$ as
\begin{equation}\label{eq:corrP-E}
\begin{aligned}
 D_{\alpha, N}(s,t):=&\langle P(s) P(t)\rangle_{N}=\frac{d^2}{dtds}C_{\alpha, N}(s,t),\\
 E_{\alpha, N}(t):=&\tfrac 12D_{\alpha, N}(t,t)\ ,\\
U_{\alpha, N}(t):=&\tfrac12D_{\alpha, N}(t,t)+\tfrac{\Omega^2}2 C_{\alpha,
N}(t,t).
\end{aligned}
\end{equation}

To compute $C_{\alpha,N}(s,t)$ we first solve the Hamilton equation for the Hamiltonian \eqref{eq:Ham1} via their
Laplace transform. From such solution it is possible to obtain an explicit expression for the Laplace transform $\tilde
C^1_{\alpha,N}(\lambda)$ of $C_{\alpha,N}(0,t)$. It is not easy to study directly the inverse Laplace transform of  $\tilde
C^1_{\alpha,N}(\lambda)$. Thus we first compute $\tilde C^1_{\alpha}(\lambda)=\lim_{N\to\infty}\tilde
C^1_{\alpha,N}(\lambda)$ and take its inverse Laplace transform obtaining the effective correlation $C_{\alpha}(0,t)$ for
the $N=\infty$ system. It is now possible to study in details the long time behavior of $C_{\alpha}(0,t)$.  We then
express the full correlation function $C_{\alpha,N}(s,t)$ in term of products and convolutions of functions depending
only on $t$ or $s$ whose Laplace transform is closely related to $\tilde C^1_{\alpha,N}(\lambda)$. This allow us to
extend the analysis to the full correlation function $C_{\alpha,N}(s,t)$ and its limit $C_{\alpha}(s,t)$.

Since $C_{\alpha,N}(0,0)=T_P/\Omega^2$ for every $N$, we can expect that $C_{\alpha,N}(s,t)$ and
$C_{\alpha}(s,t)$ stay close for short times. On the other hand since the Hamiltonian \eqref{eq:Ham1} is harmonic,
$C_{\alpha,N}(s,t)$ is a quasi periodic function. Thus we cannot expect that $C_{\alpha,N}(s,t)$ and $C_{\alpha}(s,t)$
stay close uniformly for all $s$ and $t$. A natural time scale for such a comparison is provided by the inverse of the
spacing between the $\omega_j$. We first show that, for $N$ large and  times $s$ and $t$ short when compared to
$N/\tilde \omega$, $C_{\alpha,N}(s,t)$ is well approximated by $C_\alpha(s,t)$. This is the content of our first
theorem.

\begin{theo}\label{th:appo}
Let $C_{\alpha,N}(s,t)$ be the correlation function defined in \eqref{eq:corr} for the evolution generated by the 
Hamiltonian \eqref{eq:Ham1} with probability density \eqref{eq:dens_prob} and let
\begin{equation}\label{eq:CNi}
 C_\alpha(s,t):=\lim_{N\to\infty} C_{\alpha,N}(s,t)\ ,
\end{equation}
Then there exist constants $k,K>0$ such that
\[
\left |C_{\alpha,N}(s,t)-C_\alpha(s,t)\right|\leq \alpha^2
K\left(\left(\frac{k\tilde\omega \max(s,t)}{N}\right)^{4N}+ t^2s^2e^{-kN}\right)\, .
\]
\end{theo}

\begin{rmk}\label{rmk:Aa}
\emph{In the following we will use the letters $K$ and $k$ to indicate generic constants independent of $\alpha$ and
$N$. They are not supposed to have a fixed value even when they appear in the same formula multiple times. See also
Remark \ref{rmk:kt} below.}
\end{rmk}

As we will see below, the effects of the presence of the interaction of strength $\alpha$
on the evolution of the probe are felt on a time scale of the order of $\alpha^{-2}$. For this reason, we will assume
that $N\gg\alpha^{-2}$. In this way $C_{\alpha,N}(s,t)$ and $C_{\alpha}(s,t)$ are practically indistinguishable up to
times $s$ and $t$ much longer than $\alpha^{-2}$, that is long enough to see the effect of the interaction. An analysis
similar to that leading to Theorem~\ref{th:appo} tell us that, calling
$D_\alpha(s,t)=\lim_{N\to\infty}D_{\alpha,N}(s,t)$, see \eqref{eq:corrP-E}, we have
\[
\left |D_{\alpha,N}(s,t)-D_\alpha(s,t)\right|\leq \alpha^2
K\left(\left(\frac{k\tilde\omega\max(s,t)}{N}\right)^{2N}+t^2s^2e^{-kN}\right)\,\ .
\]

\red{Theorem \ref{th:outband} and Theorem \ref{th:inband} below contain our results for the long time behavior of
$C_\alpha(s,t)$ in the non resonant and resonant case respectively. Equations \ref{eq:Cdstima} and \eqref{eq:Cdstimain}
in the theorems are formulated in term of a main term plus corrections of higher order in $\alpha$. Since these
corrections do not vanish for long times, the distinction become empty if $\alpha$ is too large. Thus, notwithstanding
our results and methods are not based on weak limit or a perturbative scheme on $\alpha$, we will be mainly interested
in the case when $\alpha$ is small. Moreover we observe that to maintain the distinction between resonant and non
resonant systems, $\alpha$ must be much smaller than $|\Omega-\mu_-|$ and $|\Omega-\mu_+|$ so that $\overline \Omega$,
see \eqref{eq:Ham1.1}, remain ``well inside'' or ``well outside'' the resonant region. Indeed one can see from
\eqref{eq:Ham1.1},  see also \eqref{eq:polone}, that the behavior of the system changes when $\overline \Omega$ crosses
$\mu_+$ or $\mu_-$ and the details of the transition are rather involved. Clearly one could analyze the cross over
situation when $\Omega\simeq\mu_\pm$ but this is outside the scope of this paper.}

Observe that, for $\alpha=0$ and any $N$, we have
\[
 C_{0,N}(s,t)=C_0(s,t)=\frac{T_P}{\Omega^2}\cos(\Omega(t-s))\, .
\]
On heuristic grounds we expect that in the non resonant case, when the unperturbed frequency of the probe is not found
inside the frequency spectrum of the bath, that is  $\Omega<\mu_-$ or $\Omega>\mu_+$, the effective interaction between
probe and bath is weak. This is summarized in our next theorem where we show that, in this case, the behavior of
$C_\alpha(s,t)$ is close to that of $C_0(s,t)$ uniformly in $t$ and $s$.

\begin{theo}\label{th:outband} Let $C_\alpha(s,t)$ be defined in
\eqref{eq:CNi} with $\Omega\not\in[\mu_-,\mu_+]$ then for $\alpha$ small enough we have
\begin{equation}\label{eq:Cdstima}
 C_\alpha(s,t)=\frac{T_P}{\Omega^2}
\cos(\Omega(\alpha)(t-s))+\alpha K(s,t)\ ,
\end{equation}
where $K(s,t)$ is a bounded function while $\Omega(\alpha)=\sqrt{\Omega^2+\alpha\gamma^{-1}}+O(\alpha^2)$.
\end{theo}

\begin{rmk}\label{rmk:kt}
\emph{In the same spirit of Remark \ref{rmk:Aa}, we will use the notation $K(t)$ or $K(s,t)$ to indicate generic
functions of $t$ or $t$ and $s$, uniformly bounded in $t$, $s$ and $\alpha$. }
\end{rmk}

\begin{rmk}\label{rmk:1time} \emph{In Section \ref{sec:decay1} we will show that, when $s=0$, we have
\begin{equation}\label{eq:C1nr}
C_\alpha(0,t)=\frac{T_P}{\Omega^2}(1-\alpha^2r_1(\alpha))
\cos(\Omega(\alpha)t) +\alpha^3 r_2(\alpha)
\cos(\rho(\alpha)t)+\alpha K_{v}(t)\ ,
\end{equation}
where $\rho(\alpha)=\mu_++O(\alpha^2)$,  $|r_1(\alpha)|$, $|r_2(\alpha)|\leq K$ while
\[
|K_v(t)|\leq\frac{K}{1+t(1+\alpha \sqrt t)}.
\]
This finer expression will be useful in Section \ref{sec:decay2} to study the long time behavior of the full correlation
function $C_\alpha(s,t)$. In Appendix \ref{app:stime} we show that also the term $K(s,t)$ in \eqref{eq:Cdstima} can be
expressed in term of oscillation of frequencies that are combinations of $\rho(\alpha)$ and $\Omega(\alpha)$ plus 
terms that decay as a power law in $t$ and/or $s$.}
\end{rmk}


Theorem \ref{th:outband} tells us that, if $\Omega$ is not in close resonance with the frequencies in the bath, then the
interaction between probe and bath is weak and remains weak for very long time.  Moreover a similar analysis gives
\[
 D_\alpha(s,t)=T_P
\cos(\Omega(\alpha)(t-s))+\alpha K(s,t)\ ,
\]
so that $E_\alpha(t)=T_P+\alpha K(t)$. Thus the temperature of the probe (or better its average kinetic energy
$E_\alpha(t)$) and its average internal energy $U_\alpha(t)$ stay close to their initial values.
Notwithstanding this, the term $K(t)$ contains oscillating terms that do not vanish in time plus decaying term that
vanish only as a power law in time. Thus the probe reaches very slowly a state in which most of the energy is still
concentrated on a oscillation with frequency $\Omega(\alpha)=\Omega+O(\alpha)$. The remaining energy
is found on oscillations with a frequency $\rho(\alpha)=\mu_++O(\alpha^2)$, or combination of $\rho(\alpha)$ and 
$\Omega(\alpha)$, with amplitudes at most $O(\alpha^3)$, see \eqref{eq:Csplit}, \eqref{eq:CNsplit}, Remark 
\ref{rmk:1time} and Appendix \ref{app:stime}.

More interesting is the situation when $\Omega$ is found in the frequency spectrum of the bath, and we have strong
effective interaction between the two, that is in the resonant case. In this situation we expect the probe
to thermalize with the bath and reach equilibrium at the temperature of the bath. Moreover we expect it to be
found in a state very close to the steady state of a probe interacting with a stochastic thermostat. This means in 
particular that, for 
large $s$ and $t$, $C_\alpha(s,t)$ decays exponentially in $t-s$. Thus a natural guess is that, for large $t$ and $s$ 
we have
\[
 C_\alpha(s,t)\simeq
 \frac{T_B}{\Omega(\alpha)^2}\cos(\Omega(\alpha)(t-s))
 e^{-\xi(\alpha)|t-s|}\ ,
\]
for suitable $\Omega(\alpha)=\Omega+O(\alpha)$ and $\xi(\alpha)=O(\alpha^2)$, where $T_B$ is the temperature of the bath, see\eqref{eq:dens_prob}.

\begin{theo}\label{th:inband} Let $C_\alpha(s,t)$ be defined in \eqref{eq:CNi}
 with $\Omega\in[\mu_-,\mu_+]$ then, for $\alpha$ small enough we have
\begin{equation}\label{eq:Cdstimain}
\begin{aligned}
 C_\alpha(s,t)=&\frac{(T_P-T_B)
e^{-\xi(\alpha)(t+s)}+T_B\, e^{-\xi(\alpha)|t-s|}}{\Omega(\alpha)^2}
\cos(\Omega(\alpha)(t-s))+\alpha K(s,t)\ ,
\end{aligned}
\end{equation}
where
\[
\Omega(\alpha)=\sqrt{\Omega^2+\alpha\gamma^{-1}}+O(\alpha^2)\qquad\hbox{and}\qquad
\xi(\alpha)=\frac{\alpha^2}{2\Omega\sqrt{(\Omega-\mu_-)(\mu_+-\Omega)}}+O(\alpha^3)\, .
\]
\end{theo}

\begin{rmk}\label{rmk:1timeres}
\emph{As for Remark \ref{rmk:1time}, we will show in Section \ref{sec:decay1} that
\begin{equation}\label{eq:C1r}
C_\alpha(0,t)=\frac{T_P e^{-\xi(\alpha)t}}{\Omega^2}(1-\alpha^2r_1(\alpha))
\cos(\Omega(\alpha)t+\phi(\alpha)) + \alpha^3
r_2(\alpha)\cos(\rho(\alpha)t)+\alpha K_v(t)\ ,
\end{equation}
with $r_1(\alpha)$, $r_2(\alpha)$, $\rho(\alpha)$ and $K_v(t)$ as in Remark \ref{rmk:1time} while
$\phi(\alpha)=O(\alpha^2)$. This characterization will be useful to obtain \eqref{eq:Cdstimain} in Section \ref{sec:decay2}.}
\end{rmk}

An analysis analogous to the one leading to Theorem \ref{th:outband} gives
\[
\begin{aligned}
E(t)=&\left(T_B+(T_P-T_B)e^{-2\xi(\alpha)t}\right)
+\alpha K(t)\ .
\end{aligned}
\]
This shows that, on a time scale of the order of $\alpha^{-2}$ the temperature (or better, the average kinetic energy) 
of the probe converges to the temperature of the bath exponentially fast in agreement with Newton's law of cooling with 
heat transfer coefficient $2\xi(\alpha)$.

A better analysis of this thermalization, developed in Appendix \ref{app:stime} shows that we can write
\begin{equation}\label{eq:poly}
C_\alpha(t,s)=C_\alpha^{th}(t-s)+\alpha^5 K_o(t+s)+\alpha K_v(t,s)
\end{equation}
where $K_v(t,s)$ vanishes as a power law when $\min(t,s)\to\infty$ while $K_o(t)$ consists of oscillations of frequency 
$\rho(\alpha)$. Thus we can say that, but for correction $O(\alpha^5)$ we have $\lim_{\tau\to\infty} 
C_\alpha(\tau,t+\tau)=C_\alpha^{th}(t)$, see \eqref{eq:nonth}. Still from Appendix \ref{app:stime} we further learn that
\begin{equation}\label{eq:thermalize}
C_\alpha^{th}(t)=\frac{T_B}{\Omega(\alpha)^2}\cos(\Omega(\alpha)t)
e^{-\xi(\alpha)|t|}+\alpha^2K_v(t)+\alpha^3 K_o(t)
\end{equation}
where again $K_v(t)$ vanishes as a power law when $t\to\infty$ while $K_o(t)$ consists of
oscillation of frequency $\rho(\alpha)$. The first term on the r.h.s. of \eqref{eq:thermalize}
can be thought as the correlation function of a stationary and mixing Markov process. In this
sense we can say that, at order 0 in $\alpha$, the probe fully thermalizes with the bath.
Notwithstanding this, the higher order corrections in $\alpha$ do not vanish even when the bath is
effectively infinite.

\begin{rmk}[\textbf{Discussion on the involved physical parameters}]
\emph{ It is worthwhile to discuss briefly the role of each of the physical parameters entering the model and how they
affect the results. The frequencies involved are, as expected, the proper frequency $\Omega=\sqrt{G/M}$ of the probe and
those of the bath $\mu_-$ and $\mu_+$, related to the on--site pinning potential and to the nearest--neighbor coupling
as it is well known. In the small $\alpha$ regime, however, we notice that the proper frequency of the oscillator is
shifted to $\Omega(\alpha)$, whose first order correction is proportional to $\alpha\gamma^{-1}=\alpha'/M$, while an
additional oscillation occurs at $\rho(\alpha)\approx\mu_+$, whose smallest degree correction is proportional to the
square of $\alpha\gamma=\alpha'/m$ (see \eqref{eq:poletto} below). We point out, moreover, that the parameter $\gamma$
affects the value of $r_2(\alpha)$ in first approximation, but not that of $r_1(\alpha)$, nor $\phi(\alpha)$ or
$\xi(\alpha)$ (see Sec.~\ref{subsec:comp} and Th.~\ref{th:inband}). }
\end{rmk}

As we already observed, since $N\gg1$ one can expect that the state of the bath will be essentially unchanged by the
interaction with the probe, even in the resonant case. This should allow us to describe the system as a probe
interacting with a stationary stochastic process. We briefly explore this idea in the following subsection.

\subsection{Stochastic Thermostat} \red{The equation of motion for $Q(t)$ can also be expressed as the solution of an
integro-differential equation involving only $Q(s)$, with $s\leq t$, and an external forcing $F(t)$. This is the content
of the following Lemma, whose proof is as easy consequence of \eqref{eq:sol}.}

\red{\begin{lemma}
The solution $Q(t)$ corresponding to the Hamiltonian \eqref{eq:Ham1} satisfies the equation
\begin{equation}\label{eq:int-diff}
    \ddot Q(t)+\overline\Omega^2Q(t)+\alpha\gamma\int_0^t dt' \left(\ddot Q(t')+\Omega^2 Q(t')\right) \Gamma_N(t-t') =\alpha F_N(t)\ ,
\end{equation}
where
\begin{equation}\label{eq:forzante-diss}
F_N(t) = \frac{1}{\sqrt N}\sum_{j=0}^N \eta_j\left(q_j(0)\cos(\omega_j t)+\frac{p_j(0)}{\omega_j}\sin(\omega_j 
t)\right)\ , \quad\hbox{and}\qquad \Gamma_N(t)=\frac1N \sum_{j=0}^N\eta_j^2\frac{\sin(\omega_j t)}{\omega_j}\ .
\end{equation}
\end{lemma}
}

\red{ Notice here that the forcing is given by the unperturbed ({\sl i.e.} $\alpha=0$) motion of the bath. 
Moreover the dissipating integral term and the forcing satisfy
\begin{equation}\label{eq:FT}
\langle F_N(t)F_N(t')\rangle = -T_B\left(1+\int_0^{t-t'}\Gamma_N (s)\,ds\right)\ .
\end{equation}
The relationship in \eqref{eq:FT} can be seen as the analog of the classical fluctuation--dissipation theorem for our 
case.\footnote{\red{Equation \eqref{eq:int-diff} and the fluctuation--dissipation
relations should be compared with the modified Langevin equation for the system in its usual form (see \cite{Ku}).
Observe though that to obtain such an equation one needs to modify the initial distribution for the bath and change the definition of the forcing $F_N$ in \eqref{eq:forzante-diss}, see \cite{Weiss}.}} In order to better understand 
why the dissipation
occurs only in the resonant case, we consider briefly, instead of \eqref{eq:int-diff}, a system where the integral term 
is replaced by an instantaneous dissipation, thus providing a much simpler picture, at the price of violating the 
fluctuation--dissipation relation. This means considering the equation}
\begin{equation}\label{eq:stocastica}
\ddot Q(t)+\overline\Omega^2Q(t) +2\dam \dot Q(t) = \alpha F_N(t)\ ,
\end{equation}
where $\alpha>0$, while the damping constant $\dam\ge 0$ is so chosen as to have the energy of the probe stay finite for
$t, s\to \infty$. The behavior of $\dam$ as a function of $\alpha$ has to be determined imposing that the correlations
of the solutions are uniformly bounded in $\alpha$. \red{We are particularly interested in the case in which the
forcing is given by \eqref{eq:forzante-diss} and the initial values $q_j(0)$ and $p_j(0)$ are distributed according to
\eqref{eq:dens_prob}.}

Taking the limit when $N\to \infty$, under suitable regularity properties, the time  correlations of the limiting 
Gaussian process $F(t)=\lim_{N\to\infty}F_N(t)$ obey
\begin{equation}\label{eq:correl_stocastica}
\langle F(t)F(s)\rangle =
\int_{\mu_-}^{\mu_+}g(\omega)\cos\left(\omega(t-s)\right)\,d\omega\ ,
\end{equation}
\red{for a suitable $g(\omega)$. We observe that, in the case where $F(t)$ is the limit of \eqref{eq:forzante-diss}, we 
get }
\begin{equation}\label{eq:Gtildeg}
g(\omega)= \frac{1}{\red{\omega^2\sqrt{(\mu_+-\omega)(\omega-\mu_-)}}}
\ .
\end{equation}

We will study the correlations 
\begin{equation}\label{eq:Cs}
C_{\alpha,\dam,N}(s,t)=\langle Q(t)Q(s)\rangle\ ,
\end{equation}
for the system \eqref{eq:stocastica}, where $\langle\cdot\rangle$ represents now the average with respect to the 
forcing $F$ and initial
distribution on $Q$ and $P$, see \eqref{eq:stocastica}. We will mostly focus on 
the limit for $N\to \infty$ setting
\begin{equation}\label{eq:Csinf}
C_{\alpha,\dam}(s,t):=\lim_{N\to\infty}C_{\alpha,\dam,N}(s,t)\, .
\end{equation}
We will not discuss the convergence to this limit, which is similar to the
Hamiltonian case, but only its limiting value \red{for $N\to \infty$, and how it depends on $\mu_-$, $\mu_+$ and on the function $g(\omega)$}. Moreover,
to compare with the results in Section \ref{sec:model}, we will assume that $\alpha$ and $\beta$ are small.

\begin{rmk}\label{rmk:Markov}
\emph{Formula \eqref{eq:correl_stocastica} for the correlations of the forcing shows that the Gaussian process $F(t)$ is
stationary and its correlation cannot decay exponentially, as they are the Fourier transform of a non--analytic
function. This implies that the forcing cannot be a Markov process, apart for the limiting case when $\mu_-\to 0$ and
$\mu_+\to \infty$ (where, for $g$ constant, it represents a white noise).}
\end{rmk}

As in Section \ref{sec:model}, the results of our analysis depend on whether we are in a non resonant case
($\overline\Omega\not\in[\mu_-,\mu_+]$) or resonant case ($\overline\Omega\in[\mu_-,\mu_+]$). In the non resonant case, when $N=+\infty$,
the contribution of the stochastic forcing to $C_{\alpha,\dam}(s,t)$ vanishes with $\alpha$, uniformly in $s$ and $t$,
for any $\dam\ge 0$. In this sense we can say that the probe does not thermalize.

In the resonant case, choosing $\dam\propto\alpha^2$, a contribution appears that stays bounded away for 0 for vanishing
$\alpha$, while it diverges for any fixed $\alpha$, when $\dam\to 0$. This contribution can be interpreted as the 
probe thermalizing with a thermostat with temperature $T_B=\alpha^2\pi g(\Omega_\dam)/4\dam$, where $\Omega_\dam^2=\overline\Omega^2-\dam^2$

\begin{theo}\label{th:stoc} 
Let $0\leq\dam<\overline\Omega$ be such that
$\dam\le\max_\pm | \overline\Omega-\mu_\pm|/2$ and assume that \red{$g(\omega)\sqrt{(\mu_+-\omega)(\omega-\mu_-)}$} is analytic for $\mu_-\le \Re \omega\le \tilde \mu_+$ and $|\Im
\omega|\le 1$ . We have that the correlation functions 
$C_{\alpha,\dam}(s,t)$, see \eqref{eq:Csinf}, satisfy
\begin{equation}\label{eq:Cdstimastoc}
\begin{split}
C_{\alpha,\dam}(s,t)=&\frac{T_Pe^{-\dam(t+s)}}{\Omega_\dam^2}\left(
\cos(\Omega_\dam(t-s))-\frac{\beta^2}{\Omega_\dam^2}\cos(\Omega_\dam (t+s)) 
+\frac{\beta\Omega_\dam}{\Omega^2}\sin(\Omega_\dam(t+s))\right) \\ &+\frac{\alpha^2\pi}{4\dam\Omega_\dam^2}
I_\mu(\Omega_\dam)g(\Omega_\dam)\left(e^{-\dam|t-s|}-e^{-\dam(t+s)}\right)\cos\left(\Omega_\dam(t-s)\right)
+\alpha^2K(s,t)\ ,
\end{split}
\end{equation}
where $I_\mu(\overline\Omega)=1$  if $\overline\Omega\in(\mu_-,\mu_+)$, 0 otherwise while $K(s,t)$ is uniformly bounded in $\alpha$, 
$\dam$, $s$, and $t$.
\end{theo}

\begin{rmk}\label{rmk:stoc_senza_smorzamento}
\emph{In the non resonant case, \eqref{eq:Cdstimastoc}  is valid also for $\beta=0$, i.e.,
\[
 C_{\alpha,0}(s,t)=\frac{T_P}{\overline\Omega^2}\cos(\overline\Omega(t-s)) + \alpha^2K_\delta(s,t)\ .
\]}
\end{rmk}

\begin{proof}
  \red{We solve \eqref{eq:stocastica} via Duhamel formula and take expectations with respect to initial data, taking into account \eqref{eq:correl_stocastica}, so that we get}
\begin{equation}\label{eq:contr}
  \begin{split}
 C_{\alpha,\dam}(s,t)&=\frac{T_Pe^{-\dam(t+s)}}{\Omega_\dam^2}\left(
 \cos(\Omega_\dam(t-s))-\frac{\beta^2}{\Omega^2}\cos(\Omega_\dam
 (t+s))  +\frac{\beta\Omega_\dam}{\Omega^2}\sin(\Omega_\dam(t+s))\right) \\
&+ \alpha^2\int_{\mu_-}^{\mu_+}d\omega g(\omega)\int_0^td\tau\int_0^sd\sigma 
 \cos(\omega(t-s-\tau+\sigma))e^{-\dam\tau}\cos(\Omega_\dam \tau) e^{-\dam\sigma}\cos(\Omega_\dam
 \sigma)\ .
  \end{split}
\end{equation}
In the non resonant case it is easy to see that the contribution in the second line of \eqref{eq:contr} is of order 
$\alpha^2$ uniformly in $\dam$, $t$, and $s$, see for example the discussion in Subsection \ref{subsec:Ctnr} below. In 
the resonant case, we can apply Corollary~\ref{cor:duhamel} to complete the proof of the theorem.
\end{proof}

\begin{rmk}\label{rmk:stoc_rumore_bianco}\emph{
Notice that the case of the white noise forcing (corresponding to $\mu_-\to 0$, $\mu_+\to \infty$ and constant $g$)
can be seen as a special case of the previous formula, where the integral over $\omega$ provides  a Dirac
delta term.}
\end{rmk}

\section{Solutions for the equation of motion: the Laplace transform}\label{sec:solution}
 
In order to find solutions for the evolution equations in a compact way, we write $q_{N+1}$, $p_{N+1}$ and
$\omega_{N+1}$ for, respectively, $Q$, $P$ and $\Omega$, and introduce the vectors $x_i=(q_i,p_i)$, for
$i=0,\ldots,N+1$, as well as the vector $X=\{x_i\}_{i=0,\ldots,N+1}$. Then the equations of motion have the compact
form
\begin{equation}\label{eq:evo}
\dot X= L X-\alpha  B X \ ,
\end{equation}
where we have introduced the $(2N+4)\times(2N+4)$ matrices $L$ and $B$. In order to simplify computations, we write such
matrices as composed by $(N+2)\times(N+2)$ square blocks $L_{ij}$, $B_{ij}$ of size $2\times 2$, indexed by $i,j$
ranging from  0 to $N+1$. We have then
\[
L_{ij}=\delta_{ij} \begin{pmatrix} 0 & 1\\
-\omega^2_i & 0 \end{pmatrix}\ ,\quad B_{ij}=\e_i\e_j\eta_i\eta_j
\begin{pmatrix} 0 & 0\\
1 & 0 \end{pmatrix}\ ,
\]
where we introduced the factor
\[
\e_i=\left\{\begin{array}{cc}1/\sqrt \gamma & i=N+1\\
-\sqrt{\frac{\gamma}{N}}& i\neq N+1\end{array}\right.\ .
\]
while $\eta_i$ is defined in \eqref{eq:normal} for $i\leq N$ with $\eta_{N+1}=1$. An implicit solution is given by
\begin{equation}\label{eq:impl}
X(t) = G(t) X(0) -\alpha\int_0^t dt'\,G(t-t')
B X(t')\ ,
\end{equation}
with $A$ denoting the block diagonal matrix
\[
G_{ij}(t) = \delta_{ij} \begin{pmatrix} \cos(\omega_i t) &
\frac{1}{\omega_i}\sin(\omega_i t)\\
-\omega_i \sin(\omega_i t)& \cos(\omega_i t) \end{pmatrix}\ .
\]

We will be mainly interested in the long term behavior of the solutions, so that we pass to the Laplace transform of
each term in \eqref{eq:impl}, getting
\[
\tilde X(\lambda) = \tilde G(\lambda) X(0) -\alpha \tilde G(\lambda)B
\tilde X(\lambda)\ ,
\]
where
\[
\tilde F(\lambda) = \int_0^{+\infty} e^{-\lambda t} F(t) \,dt,
\]
for each component of the matrix $F$. The latter equation can readily be solved with respect to $\tilde X(\lambda)$,
obtaining
\[
\tilde X(\lambda) =\left(1+\alpha \tilde G(\lambda)B\right)^{-1}
\tilde G(\lambda) X(0)=  \left(\tilde G(\lambda)^{-1}+\alpha B\right)^{-1} X(0)\, .
\]
Observe that $B=VW^T$ with $V_{2i+1}=\eta_i\e_i$ and  $V_{2i}=0$ while $W_{2i+1}=0$ and  $W_{2i}=\eta_i\e_i$, for
$i=0,\ldots, N+1$. To compute $\left(\tilde G(\lambda)^{-1}+\alpha B\right)^{-1}$ we must solve the equation
\[
(\tilde G(\lambda)^{-1}+\alpha VW^T) X=Y\ ,
\]
that we can write as
\[
 X=\tilde G(\lambda)Y-\alpha \tilde G(\lambda) VW^T X\, .
\]
Multiplying both sides by $W^T$ gives an equation for $W^T X$ that leads immediately to
\[
\left(\tilde G(\lambda)^{-1}+\alpha B\right)^{-1}=\tilde G(\lambda)
\left(I-\frac{\alpha B \tilde G(\lambda)  }{1+\alpha
W^T\tilde G(\lambda) V}\right)\ ,
\]
where $I$ is the $2(N+2)\times2(N+2)$ identity matrix. Hence we get the solution
\begin{equation*}
\begin{split}
\left(\tilde X(\lambda)\right)_i
&=\tilde G_{ii}(\lambda)\sum_{l=0}^{N+1}\left(\delta_{il}-\alpha
\frac{
B_{il} \tilde G_{ll}(\lambda)}{1+\alpha\bar
f_N(\lambda)}\right)X_l(0)\\
&= \frac{1}{\lambda^2+\omega_i^2} \sum_{l=0}^{N+1}\left(\begin{pmatrix}\lambda &
1\\ -\omega_i^2 & \lambda \end{pmatrix}\delta_{il}-\alpha
\frac{\e_i\e_l\eta_i\eta_l}{(\lambda^2+\omega_l^2)(1+\alpha\bar
f_N(\lambda))}\begin{pmatrix}\lambda& 1\\\lambda^2 &
\lambda\\
\end{pmatrix}\right)X_l(0)\ ,
\end{split}
\end{equation*}
where
\[
 \bar f_N(\lambda)=\sum_{l=0}^{N+1}
  \frac{\e^2_{l}\eta_{l}^{2}}{\lambda^2+\omega_l^2}\ .
\]
This solution can be expressed in a more explicit form as
\begin{equation}\label{eq:sol}
\begin{split}
\tilde Q(\lambda)
&=\frac{1}{D_{\alpha,N}(\lambda)}\left(\left(1+\alpha\gamma f_N(\lambda)\right)
\left(\lambda Q(0)+P(0)\right) +\frac{\alpha}{\sqrt N} \sum_{l=0}^N
\frac{\eta_l(\lambda q_l(0)+p_l(0))}{\lambda^2+\omega_l^2}\right)\ ,\\
\tilde P(\lambda)
&=\frac{1}{D_{\alpha,N}(\lambda)}\left(\left(1+\alpha\gamma f_N(\lambda)\right)
\left(\lambda P(0)-\Omega^2 Q(0)\right)-\alpha\gamma Q(0) +\frac{\alpha}{\sqrt
N} \sum_{l=0}^N
\frac{\eta_l(\lambda q_l(0)+p_l(0))}{\lambda^2+\omega_l^2}\right)\ ,
\end{split}
\end{equation}
where we have introduced the functions
\begin{equation}\label{eq:f_e_D}
\begin{aligned}
f_N(\lambda)=&
\frac{1}{N}
\sum_{j=0}^N \frac{\eta_j^{2}}{\lambda^2+\omega_j^2}\ ,\\
D_{\alpha,N}(\lambda)
=& \left(\lambda^2+\Omega^2\right)\left(1+\alpha\gamma
f_N(\lambda)\right) + \alpha\gamma^{-1}\ ,
\end{aligned}
\end{equation}
and we have used that $\tilde P(\lambda)=-Q(0)+\lambda \tilde Q(\lambda)$. Analogous expressions for $\tilde q_i$ and
$\tilde p_i$ can be obtained but would not be needed in the following.

\section{The function \matht{$f_N(\lambda)$} and its limit as \matht{$N\to\infty$}}\label{sec:f}

The function $f_N(\lambda)$ contains most of the information to understand the effect of the bath on the evolution of
the probe. In this section we study the properties of the limit of $f_N(\lambda)$ for large $N$.

Since the frequencies are distributed according to \eqref{eq:frequenze_catena}, from a minimum value $\mu_-$ to a
maximum $\mu_+ = \sqrt{\mu_-^2+4\tilde\omega^2}$, we set
\begin{equation}\label{eq:sets}
 \mathcal I:=\left[-i\mu_+, -i\mu_-\right]\cup
\left[i\mu_-, i\mu_+\right],\qquad\qquad \mathds{C}_r:=\mathds
C\backslash \mathcal I\ ,
\end{equation}
and observe that the limit
\[
 f_+(\lambda):=\lim_{N\to\infty}f_N(\lambda)
\]
is well defined for $\lambda\in \mathds{C}_r$ and it is obtained by replacing the sum in \eqref{eq:f_e_D} with an
integral, that is
\begin{equation}\label{eq:def_f_1}
f_+(\lambda)=\frac1{2\pi}
\int_{-\pi}^{\pi} \frac{dx}{\lambda^2+\mu_-^2+2\tilde\omega^2\left(1-\cos
x\right)}\ ,
\end{equation}
while, still for $\lambda\in \mathds C_r$,
\begin{equation}\label{eq:def_D_alpha} \lim_{n\to\infty} D_{\alpha,N}(\lambda)=
(\lambda^2+\Omega^2)(1 +\alpha\gamma f_+(\lambda))
+\alpha\gamma^{-1}:=D_{\alpha}(\lambda)\ .
\end{equation}

In the remainder of this section we will first find an exact expression for $f_+(\lambda)$ and then find an estimate of
the rate of convergence of $f_N(\lambda)$ to $f_+(\lambda)$.

\subsection{Exact expression for \matht{$f_+(\lambda)$}}\label{subsec:exact}

By changing the integration variable in \eqref{eq:def_f_1} we can write
\begin{equation}\label{eq:def_f_2}
f_+(\lambda)=  \frac{1}{2\pi i}\int_{|z|=1} \frac{z^{-1}
  \,dz}{\lambda^2+\mu_-^2+2 
  \tilde\omega^2- \tilde\omega^2(z+z^{-1})} =  \frac{i}{2\pi}\int_{|z|=1}
\frac{dz}{\tilde\omega^2(z-p_+(\lambda))(z-p_-(\lambda))}\ ,
\end{equation}
where $p_\pm(\lambda)$ are the roots of
\begin{equation}\label{poly}
z^2-\tfrac{\lambda^2+\mu_-^2+2 \tilde\omega^2}{ \tilde\omega^2}
z+1=0\, .
\end{equation}
For $\lambda$ real, we write them as
\begin{equation}\label{eq:radici}
p_\pm(\lambda) := 1+\frac{\lambda^2+\mu_-^2}{2\tilde
\omega^2}\pm\sqrt{\left(
  1+\frac{\lambda^2+\mu_-^2}{2\tilde \omega^2}\right)^2 -1}\, .
\end{equation}
Equation \eqref{eq:def_f_2} implies that $p_\pm(\lambda)$ lie on the unit circle if and only if $\lambda\in\mathcal I$,
see \eqref{eq:sets}, while, in general, $p_+(\lambda)p_-(\lambda)=1$. We can thus extend \eqref{eq:radici} to
$\lambda\in \mathds C_r$ by calling $p_+(\lambda)$, the root of \eqref{poly} with $|p_+(\lambda)|\geq1$. For $\lambda\in
\mathcal I$, $\lambda=iy$, we set $p_+(\lambda)=\lim_{\e\to 0^+} p_+(iy+\e)$ and $p_-(\lambda)=\lim_{\e\to 0^-}
p_-(iy+\e)$.\footnote{This corresponds to the fact that the imaginary part of $p_+(iy)$ has the same sign of $y$, while
	that of $p_-(iy)$ the opposite sign.} 
	
Finally, from Cauchy integral formula we get
\begin{equation}\label{eq:f_esplicita}
f_+(\lambda)= \frac{1}{\sqrt{(\lambda^2+\mu_-^2)(\lambda^2+\mu_+^2)}}\ .
\end{equation}

The behavior of $f_+$ near the imaginary axis is of particular interest. For $k$ real with $|k|< \mu_-$ or $|k|>
\mu_+$ we have
\begin{equation}\label{eq:nonb}
f_+(ik)=\frac{\sgn(\mu_--k)}
{\sqrt{(\mu_-^2-k^2)(\mu_+^2-k^2)}}\ ,
\end{equation}
while for $\mu_-<k<\mu_+$
\begin{equation}\label{eq:inb}
f_+(0^\pm+ik)=\mp\frac{i}{\sqrt{(k^2-\mu_-^2)(\mu_+^2-k^2)}}\ .
\end{equation}
and clearly $f_+(\lambda^*)=f_+(\lambda)^*$.

Observe that, calling $f_-(\lambda)=-f_+(\lambda)$, and
\[
\mathcal
F_+:=\{(\lambda,z)\,|\, z=1/f_+(\lambda)\}\ ,\qquad\qquad \mathcal F_-:=
\{(\lambda,z)\,|\, z=1/f_-(\lambda)\}\ .
\]
then $\mathcal F=\mathcal F_+\cup \mathcal F_-$ is a Riemann surface with 4 branching points of order 2 while $f_+$ and
$f_-$ form a meromorphic function $f$ on $\mathcal F$ with 4 poles of order 1.

\subsection{Comparison between \matht{$f_N(t)$} and \matht{$f_+(t)$}}\label{subsec:fNf+}

In this subsection we show that $f_+(\lambda)$ approximates $f_N(t)$ with an error that vanishes exponentially in $N$ 
for $\lambda$ away from $\mathcal I$. The analysis is based on the fact that $f_N(\lambda)$ can be seen as the
application of the trapezoidal rule with step $2\pi/N$ to compute the integral defining $f_+(\lambda)$. We can thus 
apply the standard methods to evaluate the error associated to the trapezoidal rule when the integrand is analytic, see 
for example \cite{DR1984} section 4.6.

Calling $w_j=e^{ix_j}$ with $x_j=j\pi/N$ and following \eqref{eq:def_f_2}, for $\lambda\in\mathds C_r$ we can write
\[
\begin{aligned}
f_N(\lambda)=&-\frac{1}{N}\sum_{j=-N+1}^N
\frac{w_j}{\tilde \omega^2(w_j-p_+(\lambda))(w_j-p_-(\lambda))}\\
=&\frac 1{2N\pi i}\int_{|z|=1+\epsilon}h(z,\lambda)\sum_{j=-N+1}^N
\frac{w_j}{z-w_j}dz-\frac 1{2 N \pi
i}\int_{|z|=1-\epsilon}h(z,\lambda)\sum_{j=-N+1}^N
\frac{w_j}{z-w_j}dz\ ,
\end{aligned}
\]
where $p_+$ and $p_-$ are defined after \eqref{eq:def_f_2} while
\[
h(z,\lambda)=-\frac{1}{\tilde \omega^2(z-p_+(\lambda))(z-p_-(\lambda))}
\]
and $\epsilon<\min\{1-|p_-(\lambda)|,|p_+(\lambda)|-1\}$. Choosing $\delta>|p_+(\lambda)|$, so that
$\delta^{-1}<|p_-(\lambda)|$, and observing that
\[
  \frac{w_j}{w_j-p_+(\lambda)} =1-\frac{w_{-j}}{w_{-j}-p_-(\lambda)
  }\ ,
\]
we get
\[
 \begin{aligned}
 f_N(\lambda)=&f_+(\lambda)\left(1 -
\frac 2N\sum_{j=N+1}^N \frac{w_j}{p_+(\lambda) - w_j }
\right )+\\
&\frac 1{2N\pi i}\int_{|z|=\delta}h(z,\lambda)\sum_{j=-N+1}^N
\frac{w_j}{z-w_j}dz-
\frac 1{2N\pi i}\int_{|z|=\delta^{-1}}h(z,\lambda)\sum_{j=-N+1}^N
\frac{w_j}{z-w_j}dz\ .
\end{aligned}
\]
Letting $\delta\to\infty$ we obtain
\[
\frac{ f_N(\lambda)-f_+(\lambda)}{f_+(\lambda)}=\frac {2}N\sum_{j=-N+1}^N
\frac{w_j}{p_+(\lambda)-w_j}:= 2G_N(p_+(\lambda))\ .
\]
For $|p|>1$, we can write
\begin{equation}\label{eq:fourier}
 G_N(p)=\sum_{n=1}^\infty\frac1{Np^n}\sum_{j=-N+1}^N
w_j^n=\sum_{n=1}^\infty\frac{(-1)^n}{Np^n}\sum_{j=0}^{2N-1}e^{i\pi jn/N}=
\frac{p^{-2N}}{1-p^{-2N}}.
\end{equation}
where we used that $\frac 1N\sum_{j=0}^{2N-1}e^{i\pi jn/N}=\delta_{n,2lN}$. Thus we get
\begin{equation}\label{eq:approssimazione_f}
f_N(\lambda)-f_+(\lambda)=
f_+(\lambda)\frac{1}{p_+(\lambda)^{2N}-1}.
\end{equation}

\section{The one time correlation function \matht{$C_{\alpha,N}(0,t)$}}\label{sec:decay1}

We are now ready to study the long time behavior of the one time correlation $C_{\alpha,N}(0,t)=\langle
Q(0)Q(t)\rangle$. The results of Sections \ref{sec:solution} and \ref{sec:f} give us a good control of the Laplace
transform $\tilde Q(\lambda)$ and $\tilde P(\lambda)$ of $Q(t)$ and $P(t)$. We can thus define
\begin{equation}\label{eq:C1}
\tilde C^1_{\alpha,N}(\lambda) := \int_0^\infty e^{-\lambda t} C_{\alpha,N}(0,t)\,dt =
\langle Q(0)\tilde Q(\lambda)\rangle_N\ .
\end{equation}
From \eqref{eq:sol} and the definition of the probability density \eqref{eq:dens_prob}, we get
\begin{equation}\label{eq:def_g^1_N}
	\tilde C^1_{\alpha,N}(\lambda) = \frac{\lambda T_P}{\Omega^2}\frac{(1+\alpha \gamma
		f_N(\lambda))}{\left(\lambda^2+\Omega^2\right)\left(1+\alpha\gamma
		f_N(\lambda)\right) + \alpha\gamma^{-1}}=:\frac{\lambda T_P}{\Omega^2}
	g^1_{\alpha,N}(\lambda)\ .
\end{equation}
We can then recover $C_N(0,t)$ from $\tilde C^1_{\alpha,N}(\lambda)$ via the anti-Laplace transform. That is we can
write 
\begin{equation}\label{eq:anti}
 C_{\alpha,N}(0,t)=\frac{1}{2\pi
i}\lim_{\Lambda\to\infty}\int_{\xi-i\Lambda}^{\xi+i\Lambda}\tilde
C^1_{\alpha,N}(\lambda)e^{\lambda
t}d\lambda\, .
\end{equation}
for $\xi>0$. Observe that since $\tilde C^1_{\alpha,N}(\lambda)$ has no singularities with positive real part (see also
below), the integral in \eqref{eq:anti} does not depend on $\xi$, for $\xi>0$.

\begin{rmk}\emph {The integral in \eqref{eq:anti} can only be defined as an improper integral since
$\tilde C^1_{\alpha,N}(0,\lambda)=O(\lambda^{-1})$ for $\lambda$ large. Observe though that
$C_{\alpha,N}(0,0^+)=T_P/\Omega^2$ (see \eqref{eq:dens_prob}) and $\lim_{\lambda\to\infty}\lambda \tilde
C^1_{\alpha,N}(\lambda)=T_P/\Omega^2$. Calling $H(t)$ the Heaviside function, we have that
$c_N(0,t)=C_{\alpha,N}(0,t)-H(t)T_P/\Omega^2$ is a continuous and piecewise differentiable function of $t\in\mathds{R}$
whose Laplace transform is $\tilde c_N(\lambda)=\tilde C^1_{\alpha,N}(\lambda)-\lambda^{-1}T_P/\Omega^2=O(\lambda^{-2})$
for $\lambda$ large. Thus we have
\[
 c_N(0,t)=\frac{1}{2\pi
i}\int_{\xi-i\infty}^{\xi+i\infty}\tilde
c_N(\lambda)e^{\lambda
t}d\lambda\ ,
\]
where the integral is now well defined. To avoid overburdening the notation, we will work with \eqref{eq:anti} without
explicitly indicating the limit as $\Lambda\to\infty$.}
\end{rmk}

As already observed, computing the anti-Laplace transform in \eqref{eq:anti} is made difficult by the singularities of
$f_N$, see also Remark \ref{rmk:shift} below. Taking the limit for $N\to\infty$ in \eqref{eq:anti} we can define
\begin{equation}\label{eq:antiinf}
 C_\alpha(0,t):=\frac{1}{2\pi
i}\int_{\xi-i\infty}^{\xi+i\infty}\tilde
C^1_\alpha(\lambda)e^{\lambda t}d\lambda\ ,
\end{equation}
where, for $\lambda\not\in\mathcal I$,  we set $\tilde C^1_\alpha(\lambda)=\lim_{N\to\infty}\tilde
C^1_{N,\alpha}(\lambda)$ and we obtain
\begin{equation}\label{eq:def_g_N}
\tilde C^1_\alpha(\lambda) = \frac{\lambda T_P}{\Omega^2}\frac{(1+\alpha\gamma
  f_+(\lambda))}{\left(\lambda^2+\Omega^2\right)\left(1+\alpha\gamma
f_+(\lambda)\right) + \alpha\gamma^{-1}}=:\frac{\lambda T_P}{\Omega^2}
g^1_\alpha(\lambda)\ .
\end{equation}
In this section we will first use the results in Section \ref{sec:f} on the relation between $f_N$ and $f_+$ to show
that $C_\alpha(0,t)$ approximates very well $C_{N,\alpha}(0,t)$ for $t$ shorter than $N$. We will then use our knowledge
of the function $f_+$ to obtain quantitative estimates on $C(0,t)$. As a preliminary step, we need to investigate
singularities and asymptotic behavior of $g^1_N(\lambda)$ and $g^1(\lambda)$ .

\subsection{Properties of \matht{$g^1_{\alpha,N}(\lambda)$} and
\matht{$g^1_\alpha(\lambda)$}}\label{subsec:galpha}

In this subsection we study the zeros, poles and asymptotic behavior of $g_N^1$ and $g^1$. We first look at the general
properties and then specialize our analysis to the resonant and non resonant cases separately.

We first observe that
\[
 \lim_{\lambda\to\infty} g^1_{\alpha,N}(\lambda)\lambda^2=1\qquad{\rm and}\qquad
  \lim_{\lambda\to\infty} g^1_\alpha(\lambda)\lambda^2=1
\]
and that $g^1_{\alpha,N}$ converges, as $N\to\infty$, to $g^1_\alpha$ in the space of analytic functions on
$\mathds{C}_r$.

To study the poles and zeros of $g^1_{\alpha,N}$, we observe that we can write $f_N=h_N/\bar h_N$ with
\[
 h_N(\lambda)=\frac 1N
\sum_{l=0}^N\prod_{j\not=l}(\lambda^2+\omega_j^2)\ ,\qquad\qquad \bar
h_N(\lambda)=\prod_{j}\eta_j^2(\lambda^2+\omega_j^2)\ ,
\]
so that
\[
 g^1_{\alpha,N}(\lambda)=\frac{\bar h_N(\lambda)+\alpha\gamma
h_N(\lambda)}{(\lambda^2+\Omega^2+\alpha\gamma^{-1})\bar h_N(\lambda)+
\alpha\gamma(\lambda^2+\Omega^2)h_N(\lambda)}=:\frac{s_{\alpha,N}(\lambda)}{r_{\alpha,N}(\lambda)}\, .
\]
Moreover $r_{\alpha,N}(\lambda)$ can be written as
\[
 r_{\alpha,N}(\lambda)=(\lambda^2+\Omega^2)\bar h_N(\lambda)+ \tilde
 h_{\alpha,N}(\lambda)\ ,
\]
where
\[
 \tilde h_{\alpha,N}(\lambda)=\alpha\gamma^{-1}\bar
h_N(\lambda)+\alpha\gamma(\lambda^2+\Omega^2)h_N(\lambda).
\]
Observe that $\bar h_N(i\xi)$ and $\tilde h_{\alpha,N}(i\xi)$ are real if $\xi\in\mathds R$. Moreover
$(\Omega^2-\xi^2)\bar h_N(i\xi)$ is 0 for every
$\xi\in\boldsymbol{\Omega}=\{\omega_0,\ldots,\omega_N,\omega_{N+1}:=\Omega\}$ while, still for
$\xi\in\boldsymbol{\Omega}$, we have that $\tilde h_{\alpha,N}(i\xi)$ is positive or negative depending on whether the
number of elements of $\boldsymbol{\Omega}$ smaller than $\xi$ is even or odd, respectively. Thus $r_{\alpha,N}(i\xi)$
has a zeros in each of the $N+1$ finite interval with end points on successive elements of $\boldsymbol{\Omega}$.
Finally observe that, if $\bar \xi$ is the largest element of $\boldsymbol{\Omega}$, we have $(\Omega^2-\xi^2)\bar
h_N(i\xi)\tilde h_{\alpha,N}(i\xi)<0$ for $\xi>\bar \xi$ while $(\Omega^2-\bar\xi^2)\bar h_N(i\bar \xi)=0$ and ${\rm
deg}((\Omega^2-\xi^2)\bar h_N)>{\rm deg}(\tilde h_{\alpha,N})$. Thus we have one more zero of $r_{\alpha,N}(i\xi)$ for
$\xi\in(\bar \xi,\infty)$.\footnote{This is nothing but the interlacing property for the potential of \eqref{eq:Ham1}.}
This gives $N+2$ zeros of $r_N(i\xi)$ on the positive real axis and thus $N+2$ on the negative real axis since
$r_N(i\xi)$ depends only on $\xi^2$. Observing that $r_N(\lambda)$ is a polynomial of degree $2N+4$, this implies that
all the singularities of $g^1_{\alpha,N}(\lambda)$ are on the imaginary axis. A similar argument for the zeros of
$s_{\alpha,N}(\lambda)$ tells us that they are all on the imaginary axis with one of them in each of the $N$ segments
$(\omega_j,\omega_{j+1})$, for $j=0,\ldots,N$, and one above $\omega_N=\mu_+$. Observe finally that no $\lambda$ can be
a zero of both $r_{\alpha,N}$ and $s_{\alpha,N}$, if $\alpha\not=0$.

Similarly calling $h_+(\lambda):=1/f_+(\lambda)$ we get
\[
 g^1_\alpha(\lambda)=\frac{h_+(\lambda)+\alpha\gamma}{(\lambda^2+\Omega^2+\alpha\gamma^{-1})
h_+(\lambda)+\alpha\gamma(\lambda^2+\Omega^2)}:=\frac{s_\alpha(\lambda)}{r_\alpha(\lambda)}.
\]
We observe that $s_\alpha$ and $r_\alpha$ are analytic in $\mathds C_r$ while $s_\alpha(i\xi)$ and $r_\alpha(i\xi)$ are
real for $\xi\in\mathds R\backslash i\mathcal I$. An analysis of their sign tell us that $s_\alpha(i\xi)$ has a zero for
$\xi\in(\mu_+,\infty)$. On the other hand, if $\Omega<\mu_-$, $r_\alpha(i\xi)$ has one zero for $\xi\in(\Omega,\mu_-)$
and one for $\xi\in(\mu_+,\infty)$ while, if $\Omega> \mu_+$, $r_\alpha(i\xi)$ has one zero for $\xi\in(\mu_+,\Omega)$
and one for $\xi\in(\Omega,\infty)$. Finally, if $\mu_-<\Omega<\mu_+$, $r_\alpha(i\xi)$ has one zero for
$\xi\in(\mu_+,\infty)$. Comparing with the discussion for $g^1_{\alpha,N}$ and using Hurwitz's Theorem, we see that
$g^1_\alpha$ has no other zero or pole than those listed above and their complex conjugates.

It is interesting to look for the singularities of
\begin{equation}\label{eq:g-}
g^1_{\alpha,-}(\lambda):=\frac{(1+\alpha\gamma
  f_-(\lambda))}{\left(\lambda^2+\Omega^2\right)\left(1+\alpha\gamma
f_-(\lambda)\right) + \alpha\gamma^{-1}}=\frac{(1-\alpha\gamma
  f_+(\lambda))}{\left(\lambda^2+\Omega^2\right)\left(1-\alpha\gamma
f_+(\lambda)\right) + \alpha\gamma^{-1}}\ ,
\end{equation}
that can be seen as the analytic extension of $g^1_\alpha$ on the Riemann surface $\mathcal F$, see Subsection
\ref{subsec:exact}. Combining $g^1_\alpha$ and $g^1_{\alpha,-}$ we look for solution of
\begin{equation}\label{eq:poli}
\left(\lambda^2+\overline\Omega^2\right)^2(\lambda^2+\mu_-^2)(\lambda^2+\mu_+^
2)
-\alpha^2\gamma^2(\lambda^2+\Omega^2)^2=0\ ,
\end{equation}
where we set $\overline \Omega=\sqrt{\Omega^2+\alpha\gamma^{-1}}$. Since we are interested in the $\alpha$ small regime,
we will solve \eqref{eq:poli} perturbatively.

Clearly if $\alpha=0$, $\pm i\Omega$ are solution of order 2 while $\pm i\mu_-$ and $\pm i\mu_+$ are solution of order
1. For small $\alpha$ we still have 8 solutions that can be written as $\pm i \Omega_+(\alpha)$, $\pm i
\Omega_-(\alpha)$, $\pm i \rho_+(\alpha)$ and $\pm i \rho_-(\alpha)$ where
\begin{equation}\label{eq:polone}
  \begin{split}
\Omega_\pm(\alpha)=&\overline \Omega\mp\alpha^2
\frac{f_+(i\overline\Omega)}{2\overline\Omega}+\frac{\alpha^3\gamma}{2\overline
  \Omega}f^2_+(i\overline \Omega)
\\
&+\frac{\alpha^4f_+^2(i\overline\Omega)}{4\overline
  \Omega}\left(f_+^2(i\overline \Omega)
(\mu_-^2+\mu_+^2-2\overline \Omega^2)-\frac{1}{2\overline \Omega^2}
\mp 2\gamma^2 f_+(i\overline\Omega)\right) +O(\alpha^5)\ .
\end{split}\end{equation}
while
\begin{equation}\label{eq:poletto}
\begin{aligned}
\rho_+(\alpha)=&\mu_++\frac{\alpha^2\gamma^2}{8\mu_+\tilde\omega^2}
+ O(\alpha^3)\ ,\\
\rho_-(\alpha)=&\mu_--\frac{\alpha^2\gamma^2}{8\mu_-\tilde\omega^2}+O(\alpha
^3)\ .\\
\end{aligned}
\end{equation}
Observe that, for small $\alpha$, in the non resonant case all the 8 solutions are on the imaginary axis while, in the
resonant case, $i\Omega_\pm(\alpha)$ acquires a non zero real part.

\begin{rmk}\label{rmk:shift}\emph{
Thus we see that most of the singularities of $g^1_{\alpha,N}$ are in the set $\mathcal I$ on the imaginary axis. Their
structure makes it very difficult to compute $C_{\alpha,N}(0,t)$ using \eqref{eq:anti} and shifting the integral from
$\xi>0$ to $\xi<0$. By taking the limit as $N\to\infty$ we see that $g^1_\alpha$ has 2 or 4 poles on the imaginary axis
outside $\mathcal I$ while it inherits form $f_+$ a jump discontinuity on $\mathcal I$ and square root singularities at
$\pm i\mu_-$ and $\pm i\mu_+$. Thus it will be much easier to study the behavior of $C_{\alpha}(0,t)$ using
\eqref{eq:antiinf}.}
\end{rmk}

To summarize we distinguish between the two physically relevant cases. Since $g^1_{\alpha,N}$ and $g^1_\alpha$ depend
only on $\lambda^2$, we only discuss poles in the half plane $\mathds C^+=\{z\, |\, \Im z>0\}$.

\begin{description}
\item [The non resonant case] All upper half plane poles of $g^1_{\alpha, N}$ but two are in the set $(i\mu_-,i\mu_+)$
on the imaginary axis. If $\Omega<\mu_-$ of the two remaining poles, one is in $(i\Omega,i\mu_-)$ and the other in
$(i\mu_+,i\infty)$. These poles converge to the corresponding poles of $g^1_{\alpha}$ which, for $\alpha$ small, are
given by $i\Omega_+(\alpha)\in (i\Omega,i\mu_-)$ while $i\rho_+(\alpha)\in (i\mu_+,i\infty)$. Analogously, if
$\Omega>\mu_+$, of the two remaining poles of $g^1_{\alpha,N}$, one is in $(i\mu_+,i\Omega)$, the other in
$(i\Omega,i\infty)$. Again, the corresponding poles of $g^1_\alpha$, for $\alpha$ small, are $i\rho_+(\alpha)\in
(i\mu_+,i\Omega)$ and $i\Omega_+(\alpha)\in (i\Omega,i\infty)$.

\item [The resonant case] In this case, all upper half plane poles of $g^1_{\alpha,N}$ but one are in $[i\mu_-,i\mu_+]$
and the remaining one is in $(i\mu_+,i\infty)$. This converges to the pole of $g^1_\alpha$ given by $i\rho_+(\alpha)$,
for $\alpha$ small. In this case, it is important to notice that both $\pm i\Omega_+(\alpha)$ and $\pm
i\Omega_-(\alpha)$ are poles of $g^1_{\alpha,-}$. Although they do not directly appear in $g^1_\alpha$, they are very
close to $\mathcal I$ and they will play a fundamental role in computing \eqref{eq:antiinf}.

\end{description}

\subsection{Bounds for \matht{$|C_{\alpha,N}(0,t)-C_{\alpha}(0,t)|$}}\label{subsec:bound}

We first observe that
\[
g^1_{\alpha,N}(\lambda)-g^1_\alpha(\lambda) = \alpha^2\frac{1}
{D_{\alpha,N}(\lambda)D_\alpha(\lambda)}(f_N(\lambda)-
f(\lambda))
=\alpha^2\frac{f_+(\lambda)}
{D_{\alpha,N}(\lambda)D_\alpha(\lambda)}G_N(p_+(\lambda))\ ,
\]
If we take $\lambda$ with $\Re(\lambda)>1$ we get
\[
\left| \frac{\lambda f_+(\lambda)}
{D_{\alpha,N}(\lambda)D_\alpha(\lambda)}\right|\leq\frac{K}{|\lambda|^{5}}\ ,
\]
while from \eqref{eq:radici} it follows that, still assuming $\Re(\lambda)>1$, we have $|p_+(\lambda)|\geq
k\Re(\lambda)/\tilde \omega^2$, see Remark \ref{rmk:Aa}.

Choosing $\xi=N/t$ in \eqref{eq:anti} and \eqref{eq:antiinf} we get, for $t<N$,
\begin{equation}\label{eq:N^N}
 |C_{\alpha,N}(0,t)-C_\alpha(0,t)|\leq \alpha^2
e^{N}\left(\frac{k\tilde\omega t}N\right)^{2N}
\int_{-\infty}^{\infty} \frac{K\,dx}{|N/t+ix|^5}\leq
\alpha^2 K\left(\frac{k\tilde\omega  t}N\right)^{2N}.
\end{equation}

\subsection{Asymptotic behavior of \matht{$C_\alpha(0,t)$}}\label{subsec:asym}

We can now write
\begin{equation}\label{0plus}
 C_\alpha(0,t)=\frac{1}{2\pi i}\int_{0^+-i\infty}^{0^++i\infty}\tilde
C^1_\alpha(\lambda)e^{\lambda
t}d\lambda\,.
\end{equation}

\begin{rmk}\label{rmk:00+}\emph{
Since the function $\tilde C^1_\alpha$ presents pole singularities and discontinuities on the imaginary axis the
integration path in \eqref{0plus} follows the imaginary axis, where $\widetilde C^1_\alpha(\lambda)$ is taken as
$\widetilde C^1_\alpha(\lambda+0^+)$, but for $\delta$-neighborhood of $\pm i\rho_+(\alpha)$ and, depending on $\Omega$,
of  $\pm i\Omega_+(\alpha)$, where it is replaced by the path $\pm i\rho_+(\alpha)+l_{\delta,+}(s)$ with
$l_{\delta,+}(s)=\delta e^{i s}$, $s\in [-\pi/2,\pi/2]$, or $\pm i\Omega_+(\alpha)+l_{\delta,+}(s)$, respectively.
Analogously we define $\int_{0^--i\infty}^{0^-+i\infty}$ with $l_{\delta,-}(s)=\delta e^{i s}$, $s\in [\pi/2,3\pi/2]$, in
place of $l_{\delta,+}$.}
\end{rmk}

To compute this integral we want to shift the integration path to the negative real half plane. The results of this
shift depends again on the value of $\Omega$.

\subsubsection{The non resonant case}\label{subsec:nonr}

By shifting the integral in \eqref{0plus} form $0^+$ to $0^-$, see Remark \ref{rmk:00+}, we get
\begin{equation}\label{eq:tutti}
\begin{aligned}
 C_\alpha(0,t)=&2\Res(\tilde C^1_\alpha,i\Omega_+(\alpha))\cos(\Omega_+(\alpha)t)+
2\Res(\tilde C^1_\alpha,i\rho_+(\alpha))\cos(\rho_+(\alpha)t)+\\
&C_{\alpha,d}(t)+\frac{1}{2\pi i}\int_{0^--i\infty}^{0^-+i\infty}\tilde
C^1_\alpha(\lambda)e^{\lambda t}d\lambda\ ,
\end{aligned}
\end{equation}
where $\Res(f,z)$ is the residue of the meromorphic function $f$ at the point $z$, while $C_{\alpha,d}(t)$ accounts for
the integration around the discontinuity on $\mathcal I$ and, using that $\tilde C^1_\alpha(-\lambda)=-\tilde
C^1_\alpha(\lambda)$, is given by
\begin{equation}\label{eq:Cd}
 C_{\alpha,d}(t)=\frac{1}{\pi}\int_{\mu_-}^{\mu_+}
\left(\tilde C^1_\alpha(0^++i\xi)-\tilde
C^1_\alpha(0^-+i\xi)\right)\cos(\xi t)d\xi\ ,
\end{equation}
We first observe now that for every $\xi<0$
\[
 \frac{1}{2\pi i}\int_{0^--i\infty}^{0^-+i\infty}\tilde
C^1_\alpha(\lambda)e^{\lambda t}d\lambda=
\frac{1}{2\pi i}\int_{\xi-i\infty}^{\xi+i\infty}\tilde
C^1_\alpha(\lambda)e^{\lambda t}d\lambda\ ,
\]
so that, letting $\xi\to-\infty$, the last term in the r.h.s. of \eqref{eq:tutti} vanishes.

On the other hand, for $\xi>0$, we have
\begin{equation}\label{eq:Hacca}
\begin{aligned}
\tilde C^1_\alpha(0^++i\xi)-\tilde C^1_\alpha(0^-+i\xi)=
&\frac{T_P}{\Omega^2}\frac{2\alpha^2 \xi\sqrt{(\xi^2-\mu_-^2)(\mu_+^2-\xi^2)}}
{(\overline \Omega^2-\xi^2)^2(\xi^2-\mu_-^2)(\mu_+^2-\xi^2)+\alpha^2\gamma^2(\Omega^2-\xi^2)^2}=:\\
&\frac{2\alpha^2 T_P}{\Omega^2}\sqrt{(\xi-\mu_-)(\mu_+-\xi)}\mathcal
G(\xi)\ ,
\end{aligned}
\end{equation}
so that we get
\begin{equation}\label{eq:Bessel}
\begin{aligned}
 C_{\alpha,d}(t)=&\frac{2\alpha^2 T_P}{\pi\Omega^2}\int_{\mu_-}^{\mu_+}
\sqrt{(\xi-\mu_-)\left(\mu_+-\xi\right)}\mathcal G(\xi)\cos(\xi t)d\xi=\\
&\frac{\alpha^2 T_P\delta_\mu}{\pi\Omega^2}\int_{-1}^{1}\sqrt{1-\kappa^2}\overline
{\mathcal G}(\kappa)
\cos(\bar \mu t + \delta_\mu \kappa t)d\kappa\ ,
\end{aligned}
\end{equation}
where we set $\xi=\bar \mu+\delta_\mu \kappa$, with $\bar \mu=(\mu_-+\mu_+)/2$ and $\delta_\mu=(\mu_--\mu_+)/2$, and
$\overline {\mathcal G}(\kappa)=\mathcal G(\bar \mu+\delta_\mu \kappa)$. We can thus apply Lemma \ref{lem:stima} and,
given $\e\leq \frac12$, we obtain
\begin{equation}\label{eq:optimal}
|C_{\alpha,d}(t)|\leq \sup_{\kappa\in\mathcal R_d}
\left|(1-\kappa^2)^{\frac12-\e}\overline {\mathcal G}(\kappa)\right|
\frac{K\alpha^2 T_P}{\Omega^2}
\frac 1{t^{1+\e}}\ ,
\end{equation}
where $\mathcal R_d=\{\kappa\,:\, |\Re(\kappa)|\leq 1\,,\, |\Im(\kappa)|\leq d\}$ for $d>0$. Observe that $\mathcal
G(\xi)$ has two poles of order 1 at $k=\rho_\pm(\alpha)$, with $\rho_\pm(\alpha)=\mu_\pm+O(\alpha^2)$ and
$\rho_-(\alpha)<\mu_-$ while $\rho_+(\alpha)>\mu_+$. Thus, for $\xi$ close to $\mu_\pm$ we have
\begin{equation}\label{eq:near}
 \mathcal G(\xi)\simeq \frac{K}{\xi-\rho_\pm(\alpha)}\ ,
\end{equation}
from which we obtain that for $\e>-1/2$ we have
\[
\sup_{\kappa\in\mathcal R_d}
\left|(1-\kappa^2)^{\frac12-\e}\overline {\mathcal G}(\kappa)\right|=K\alpha^{-1-2\e}\, .
\]
The optimal bound is thus obtained choosing $\epsilon=-1/2$ for $t\leq \alpha^{-2}$ and
$\epsilon=1/2$ for $t\geq \alpha^{-2}$.

From \eqref{eq:near} we also see that we can write
\[
 2\Res(\tilde C^1_\alpha,i\rho_+(\alpha))=:\alpha^3 r_2(\alpha)\ ,
\]
while calling $R(\alpha):=2\Res(\tilde C^1_\alpha,i\Omega_+(\alpha))$ we can finally write
\begin{equation}\label{eq:final_out}
 C_\alpha(0,t)=R(\alpha)\cos(\Omega_+(\alpha)t)+ \alpha^3 r_2(\alpha)
 \cos(\rho_+(\alpha)t)+ \alpha^2 K_{v}(t)\ ,
\end{equation}
where the first term gives the principal contribution, the second is an oscillating correction and
\begin{equation}\label{eq:stima1}
 |K_{v}(t)|\leq \frac{K}{1+\sqrt t(1+\alpha^2 t)}\ .
\end{equation}
is an asymptotically vanishing correction.

\begin{rmk}\label{rmk:int}\emph{Notice that form \eqref{eq:stima1} we get
\[
 \int_0^\infty |K_{v}(t)|dt\leq K\alpha^{-1}\ .
\]
In this sense, one can say the contribution $C_{\alpha,d}$ to the
correlation function is $O(\alpha)$} .
\end{rmk}

\subsubsection{The resonant case}\label{subsec:r}

Proceeding as for \eqref{eq:tutti} we get

\begin{equation}\label{eq:tutti_bis}
 C_\alpha(0,t)=
\alpha^3 r_2(\alpha) \cos(\rho_+(\alpha)t)+C_{\alpha,d}(t)\ ,
\end{equation}
with $C_{\alpha,d}(t)$ still given by \eqref{eq:Cd}. The main difference with Subsection \ref{subsec:nonr} is that the
function $\mathcal G(\xi)$ has poles at $\pm\Omega_\pm(\alpha)$ close to the integration domain $[\mu_-,\mu_+]$.

We thus proceed as for \eqref{eq:Bessel} but use Corollary \ref{cor:stimasing} and we get
\[
C_{\alpha,d}(t)=
|R(\alpha)|\cos(\Omega_p(\alpha)t+\phi(\alpha))e^{-\xi(\alpha)t}+\overline
C_{\alpha,d}(0,t)\ ,
\]
where we wrote $\Omega_+(\alpha):=\Omega_p(\alpha)+i\xi(\alpha)$ while
\[
 R(\alpha)=|R(\alpha)|e^{i\phi(\alpha)}:=\frac{4\alpha^2 T_P}{\Omega^2}
\sqrt{(\Omega_+(\alpha)-\mu_-)(\mu_+-\Omega_+(\alpha))}\Res(\mathcal
G,\Omega_+(\alpha))\ .
\]
Observe that $\mathrm{Res}(\mathcal G,\Omega_+(\alpha))=O(\alpha^{-2})$ due to the presence of the pole in
$\Omega_-(\alpha)$ while $\xi(\alpha)=O(\alpha^2)$ and $\phi(\alpha)=O(\alpha^2)$. See Subsection \ref{subsec:comp} for
more precise values.

In analogy with \eqref{eq:final_out} we can write
\begin{equation}\label{eq:final_in}
C(0,t)=|R(\alpha)|\cos(\Omega_p(\alpha)t+\phi(\alpha)) e^{-\xi(\alpha)t} +
\alpha^3 r_2(\alpha)\cos(\rho_+(\alpha)t)+\alpha^2 K_{v}(t)\ ,
\end{equation}
where $\phi(\alpha)=O(\alpha^2)$, see \eqref{eq:poloneex} below, and $K_{v}(t)$ still satisfies \eqref{eq:stima1} and
Remark \ref{rmk:int}.

\subsection{Computing the residues}\label{subsec:comp}

We are left with the task of computing the residue $R(\alpha)$ and $r_2(\alpha)$. To this extent we observe that in both
the cases studied above we can write
\[
 R(\alpha)=2\Res(\tilde C^1_{\alpha}(\lambda)-\tilde C^1_{\alpha,-}(\lambda),i\Omega_+(\alpha))
\]
where
\[
\tilde C^1_{\alpha,-}(\lambda)=\frac{\lambda T_P}{\Omega^2}
g^1_{\alpha,-}(\lambda)
\]
see \eqref{eq:g-}. That is, $\tilde C^1_{\alpha,-}(\lambda)$ is the analytic continuation of $\tilde
C^1_{\alpha}(\lambda)$ past the discontinuity at $\mathcal I$ computed using $f_-(\lambda)$ in place of $f_+(\lambda)$.
This is so because the $\tilde C^1_{\alpha}$ and $\tilde C^1_{\alpha,-}$, if $\alpha\not=0$, have no common singularity,
see Subsection \ref{subsec:galpha}. A similar identity holds for $r_2(\alpha)$. As for \eqref{eq:Hacca} we get
\[
\begin{aligned}
 \tilde C^1_\alpha(\lambda)-\tilde C^1_{\alpha,-}(\lambda)=
&\frac{T_P}{\Omega^2}\frac{2\alpha^2
\lambda\sqrt{(\lambda^2+\mu_-^2)(\lambda^2+\mu_+^2)}}{(\lambda^2+\Omega^2_+(\alpha))
(\lambda^2+\Omega^2_-(\alpha))(\lambda^2+\rho^2_+(\alpha))
(\lambda^2+\rho^2_-(\alpha))}\ ,
\end{aligned}
\]
so that
\[
R(\alpha)=\frac{T_P}{\Omega^2}
\frac{2\alpha^2\sqrt{(\mu_-^2-\Omega^2_+(\alpha))(\mu_+^2-
    \Omega^2_+(\alpha))}}{
(\Omega^2_-(\alpha)-\Omega^2_+(\alpha))(\rho^2_+(\alpha)-\Omega^2_+(\alpha))
(\rho^2_-(\alpha)-\Omega^2_+(\alpha))}\ ,
\]
and, using \eqref{eq:polone}, we get
\begin{equation}\label{eq:poloneex}
 R(\alpha)=\frac{T_P}{\Omega^2}\left(1-\alpha^2f_+^3(i\Omega)\left(\frac{\mu_-^2+
 \mu_+^2}2-\Omega^2\right)\right)+O(\alpha^3)\, .
\end{equation}
With a similar argument we get
\[
 r_2(\alpha)=-\frac{2T_P}{\Omega^2}\frac{\gamma}{\left(\mu_+^2-\Omega^2\right)^2(\mu_+^2-\mu_-^2)}+O(\alpha)\,.
\]

\section{The two time correlation function \matht{$C_N(s,t)$}}
\label{sec:decay2}

In this section we extend the analysis of Section \ref{sec:decay1} to the full two time correlation function
$C_{\alpha,N}(s,t)$. As for Section \ref{sec:decay1}, we will use the exact expression for the Laplace transform of
$\widetilde Q$ obtained in Section \ref{sec:solution}. We thus define the Laplace transform $\tilde
C_{\alpha,N}(\lambda,\lambda')$ of $C_{\alpha,N}(s,t)$ as
\[
\tilde C_{\alpha,N}(\lambda,\lambda'):=\int_0^\infty\int_0^\infty e^{-\lambda t}e^{-\lambda'
s}C_{\alpha,N}(s,t)dtds=\int_0^\infty\int_0^\infty e^{-\lambda s}e^{-\lambda'
t}\langle Q(s)Q(t)\rangle_N\,dsdt \ .
\]
Taking into account the distribution $\rho_N$ for the initial values, see \eqref{eq:dens_prob}, and the exact expression
for $\tilde Q$, see \eqref{eq:sol}, we get
\begin{equation}\label{eq:trasf_Q^2}
\begin{split}
\tilde C_{\alpha,N}(\lambda,\lambda')=&
\frac{1}{D_{\alpha,N}(\lambda)D_{\alpha,N}(\lambda')}
 \Biggl[\left(1+\alpha \gamma f_N(\lambda))
(1+\alpha\gamma f_N(\lambda')\right)
\left(\lambda\lambda'\langle Q(0)^2\rangle+\langle P(0)^2\rangle\right)\\
&\qquad\qquad\qquad\qquad
+\frac{\alpha^2}{N}\sum_j\frac{\eta_j^2(\lambda\lambda'\langle
q_j(0)^2\rangle+\langle p_j(0)^2\rangle)}
{\left(\lambda^2+\omega_j^2\right)
\left(\lambda'^2+\omega_j^2\right)}\Biggr]\\
=& T_P
g^1_{\alpha,N}(\lambda)g^1_{\alpha,N}(\lambda')\frac{\lambda\lambda'+\Omega^2}
{\Omega^2}
+\frac{\alpha^2 }{D_{\alpha,N}(\lambda)D_{\alpha,N}(\lambda')}
g^2_{N}(\lambda,\lambda')\\
=:&\tilde C_{\alpha,N}^{\rm nt}(\lambda,\lambda') +\alpha^2\tilde
C_{\alpha,N}^{\rm t}(\lambda,\lambda') \ ,
\end{split}
\end{equation}
where $g^1_N$ is defined in \eqref{eq:def_g^1_N} and
\begin{equation*}
g^2_{N}(\lambda,\lambda'):=\frac{T_B}{N}\sum_j
\frac{\eta_j^2(\lambda\lambda'+\omega_j^{2})}{\omega_j^2\left (
  \lambda^2+\omega_j^2\right)\left(\lambda'^2+\omega_j^2\right)}\, .
\end{equation*}
We will consider separately the two terms defined in \eqref{eq:trasf_Q^2}. More precisely we define
\begin{equation}\label{eq:anti2}
C^{\rm nt}_{\alpha,N}(s,t) =-\frac{1}{4\pi^2}
\int_{\xi-i\infty}^{\xi+i\infty}\int_{\xi-i\infty}^{\xi+i\infty}\tilde
C_{\alpha,N}^{\rm nt}(\lambda,\lambda')e^{\lambda s +\lambda' t}
dsdt\ ,
\end{equation}
and similarly for $C^{\rm t}_{\alpha,N}(s,t)$, so that
\begin{equation}\label{eq:Csplit}
 C_{\alpha,N}(s,t)=C^{\rm nt}_{\alpha,N}(s,t)+\alpha^2  C^{\rm t}_{\alpha,N}(s,t)\, .
\end{equation}
Although we have an explicit expression, a direct analysis of the inverse Laplace transform in \eqref{eq:anti2} is
quite
difficult. We will show that both  $C_{\alpha,N}^{\rm nt}$ and $C_{\alpha,N}^{\rm t}$ can be expressed is term of
functions depending only on $s$ or $t$ whose behavior can be analyzed using the methods of Section \ref{sec:decay1}. In
particular we will show that, in the non resonant case, $C^{\rm nt}_{\alpha,N}(t,t)$ is the dominant contribution to
$C_{\alpha,N}(t,t)$ when $t$ is large and it represent the fact that, in this case, the probe does not thermalize with
the chain. On the other hand, in the resonant case $C^{\rm t}_{\alpha,N}(t,t)$ is dominant for large $t$ and represents
the fact that the probe does thermalize with the chain, after a long enough time.

\subsection{Behavior of \matht{$C^{\rm nt}_{\alpha,N}(s,t)$}}\label{subsec:C12}

Observing that
\begin{equation*}
\tilde C^{\rm nt}_{\alpha,N}(\lambda,\lambda') = \frac{\Omega^2}{T_P}
\tilde C^1_{\alpha,N}(\lambda)\tilde C^1_{\alpha,N}(\lambda')
\left(1+\frac{\Omega^2}{\lambda\lambda'}\right),
\end{equation*}
see \eqref{eq:C1}, we can write the inverse Laplace transform of $\tilde C^{\rm nt}_{\alpha,N} $ as
\begin{equation}\label{eq:CNsplit}
C^{\rm nt}_{\alpha,N}(s,t)=\frac{\Omega^2}{T_P}C_{\alpha,N}(0,s)C_{\alpha,N}(0,t)
+\frac{\Omega^4}{T_P}\left(\int_0^s C_{\alpha,N}(0,\tau)d\tau\right)\left(\int_0^t
C_{\alpha,N}(0,\tau)d\tau\right) ,
\end{equation}
so that its contribution to the $C_{\alpha,N}(s,t)$ is completely determined by $C_{\alpha,N}(0,t)$. Thus calling
\[
S_{\alpha,N}(0,t):=\int_0^s C_{\alpha,N}(0,\tau)d\tau=\frac{T_P}{2\pi \Omega^2 i}\int_{\xi-i\infty}^{\xi+i\infty}
g^1_{\alpha,N}(\lambda)d\lambda\ ,
\]
see \eqref{eq:C1} and \eqref{eq:def_g^1_N}, we can write
\begin{equation*}
 C^{\rm nt}_\alpha(s,t):=\frac{\Omega^2}{T_P}\left(C_\alpha(0,s)C_\alpha(0,t)
+\Omega^2 S_\alpha(0,s) S_\alpha(0,t)\right)\, .
\end{equation*}
Reasoning like in Section \ref{subsec:bound} we get that, for $t<N$,
\begin{equation}\label{eq:N^NS}
|S_{\alpha,N}(0,t)-S_\alpha(0,t)|\leq 
\alpha^2 K\left(\frac{k\tilde\omega  t}N\right)^{2N}\ ,
\end{equation}
so that, for $t,s<N$ we obtain
\begin{equation}\label{eq:diff-nt}
|C^{\rm nt}_{\alpha,N}(s,t)-C^{\rm nt}_{\alpha}(s,t)|\leq \alpha^2 K\left(\frac{k\tilde\omega 
\max(s,t)}{N}\right)^{2N}\ .
\end{equation}
To analyze the behavior of $S_\alpha(0,t)$ we can repeat the argument of section \ref{subsec:asym}. The results can be
summarized as follows.

\begin{description}
\item[Non resonant case:] we get
\begin{equation}\label{eq:Sfinal_in}
S_\alpha(0,t)=\frac{R(\alpha)}{\Omega_+(\alpha)}\sin(\Omega_+(\alpha)t)+ \alpha^3
\frac{r_2(\alpha)}{\rho_+(\alpha)}\sin(\rho_+(\alpha)t)+ \alpha^2
K_{v}(t)\ ,
\end{equation}
with $K_{v}(t)$ still satisfying \eqref{eq:stima1}. Combining with \eqref{eq:final_out} we get
\begin{equation*}
C^{\rm nt}_\alpha(s,t)=\frac{\Omega^2}{T_P}R(\alpha)^2\left(\cos(\Omega_+(\alpha)s)\cos(\Omega_+(\alpha)t)+
\frac{\Omega^2}{\Omega_+(\alpha)^2}\sin(\Omega_+(\alpha)s)\sin(\Omega_+(\alpha)t)\right)+
\alpha^2 K(s,t)\ ,
\end{equation*}
where $K(s,t)$ contains oscillating corrections to the main behavior and terms that vanish as $t,s\to\infty$. Using
\eqref{eq:polone} and \eqref{eq:poloneex}, we see that
\[
\frac{\Omega^2}{T_P}R(\alpha)^2=\frac{T_P}{\Omega^2}+O(\alpha^2)\qquad\hbox{and}\qquad 
\frac{\Omega^2}{\Omega_+(\alpha)^2}=1+O(\alpha)\ ,
\]
so that we can write
\begin{equation}\label{eq:Cntnr}
 C^{\rm nt}_\alpha(s,t)=\frac{T_P}{\Omega^2}\cos(\Omega_+(\alpha)(t-s))+
 \alpha K(s,t)\ ,
\end{equation}
where the only contribution of order $\alpha$ to the correction term $K(s,t)$ is of the form
$\sin(\Omega_+(\alpha)s)\sin(\Omega_+(\alpha)t)$. It follows that in this case
\[
\limsup_{\tau\to\infty}C^{\rm nt}_\alpha(\tau,\tau+t)- \liminf_{\tau\to\infty}C^{\rm 
nt}_\alpha(\tau,\tau+t)=O(\alpha)\, .
\]

\item[Resonant case:] proceeding in a similar way in this case we get
\[
S_\alpha(0,t)=\frac{|R(\alpha)|}{|\Omega_+(\alpha)|}e^{-\xi(\alpha)t}
\sin(\Omega_p(\alpha)t+\bar\phi(\alpha))+ 
\alpha^3
\frac{r_2(\alpha)}{\rho_+(\alpha)}\sin(\rho_+(\alpha)t)+ \alpha^2
K_{v}(t)\ ,
\]
where $\bar\phi(\alpha)=\phi(\alpha)+\phi'(\alpha)$ with $\Omega_+(\alpha)=e^{i\phi'(\alpha)}|\Omega_+(\alpha)|$. 
Collecting the principal term in $\alpha$ we get
\begin{equation}\label{eq:Cntr}
C^{\rm nt}_\alpha(s,t)=\frac{T_P}{\Omega^2}\cos(\Omega_p(\alpha)(t-s))e^{-\xi(\alpha)(t+s)}+
\alpha K(s,t)\ ,
\end{equation}
where again the only contribution of order $\alpha$ to the correction term $K(s,t)$ is of the form
$\sin(\Omega_p(\alpha)s)\sin(\Omega_p(\alpha)t)$. On the other hand, we observe that the terms involving oscillations of
frequency $\rho_+(\alpha)$ are the only contribution to $C_\alpha(0,t)$ and $S_\alpha(0,t)$ that do not vanish as
$t\to\infty$, see \eqref{eq:final_in} and \eqref{eq:Sfinal_in}. We thus have that in this case
\begin{equation}\label{eq:CtMark}
\limsup_{\tau\to\infty}C^{\rm nt}_\alpha(\tau,\tau+t)- \liminf_{\tau\to\infty}C^{\rm 
nt}_\alpha(\tau,\tau+t)=O(\alpha^6)\, .
\end{equation}

\end{description}

\subsection{Behavior of \matht{$C^{\rm t}_{\alpha,N}(s,t)$}}\label{subsec:C22}

We can now come to the analysis of the second term in \eqref{eq:Csplit}.The main observation is that $C^{\rm 
t}_{\alpha,N}(s,t)$ can be written as a convolution of functions depending on a single variable. More precisely, after 
some straightforward algebra we get
\begin{equation}\label{eq:expCt}
C^{\rm t}_{\alpha,N}(s,t)=\int_0^s\int_0^t {\mathcal K_N}(t-s-\tau+\sigma)d_{\alpha,N}(\tau)d_{\alpha,N}(\sigma)d\tau
d\sigma\ ,
\end{equation}
where
\[
d_{\alpha,N}(t)= \frac{1}{2\pi
  i}\int_{\xi-i\infty}^{\xi+i\infty}\frac{ e^{\lambda t}
  d\lambda}{D_{\alpha,N}(\lambda)}\ ,
\]
with $D_{\alpha,N}(\lambda)$ defined in \eqref{eq:f_e_D}, while
\[
{\mathcal K}_N(t)=\frac{T_B}{N}\sum_{j=0}^N 
\frac{\eta_j^2\cos(\omega_jt)}{\omega_j^2}=\frac{T_B}{2N}\sum_{j=-N+1}^N
\frac{\cos(\omega_jt)}{\omega_j^2}\, ,
\]
see \eqref{eq:normal}. We can thus define
\begin{equation}\label{eq:calK}
 C^{\rm t}_{\alpha}(s,t):=\int_0^s\int_0^t {\mathcal
   K}(t-s-\tau+\sigma)d_{\alpha}(\tau)d_{\alpha}(\sigma)d\tau
 d\sigma\ ,
\end{equation}
where
\begin{equation}\label{eq:def-d-K}
 d_{\alpha}(t):= \frac{1}{2\pi i}\int_{\xi-i\infty}^{\xi+i\infty}\frac{ e^{\lambda \tau} d\lambda}{D_{\alpha}(\lambda)}
\qquad\mathrm{and}\qquad {\mathcal K}(t):=\frac{T_B}{2\pi}\int_0^{2\pi}
\frac{\cos(\omega(\theta)t)}{\omega(\theta)^2}d\theta\, ,
\end{equation}
and, analogously to what we did in Subsection \ref{subsec:C12}, we first compare $C_{N,\alpha}(s,t)$ and $C_\alpha(s,t)$
and then study the behavior of $C_\alpha(s,t)$. To start with, we show that $d_{\alpha,N}(t)$ and $d_\alpha(t)$ are 
closely related to $C_{\alpha,N}(0,t)$ and  $C_{\alpha}(0,t)$, see Section \ref{sec:decay1}.

\subsubsection{Properties of \matht{$d_{\alpha,N}(t)$ and $d_\alpha(t)$}}
\label{subsec:dalpha}
	
We observe that
\[
\frac{1}{D_{\alpha,N}(\lambda)}-\frac{1}{D_{\alpha}(\lambda)}=
\frac{\alpha\gamma(\lambda^2+\Omega^2)(f_+(\lambda)-f_N(\lambda))}{D_{\alpha,N}(\lambda)D_{\alpha}(\lambda)}\ ,
\]
so that, proceeding as in Subsection \ref{subsec:bound}, for $t<N$ we get
\[
 |d_{\alpha,N}(t) - d_\alpha(t)|\leq \alpha K\left(\frac{k\tilde\omega t}N\right)^{2N}\, .
\]
Moreover the structure of discontinuity and singularities of $1/D_\alpha(\lambda)$ is very similar to that of 
$g^1_\alpha(\lambda)$. In particular we can write
\[
\frac 1{D_\alpha(0^++i\xi)}-\frac 1{D_\alpha(0^-+i\xi)}= 
\frac{\gamma\Omega^2}{\alpha T_P}\frac {\xi^2-\Omega^2}{i\xi} \left(\tilde C^1_\alpha(0^++i\xi)-\tilde 
C^1_\alpha(0^-+i\xi)\right)\, ,
\]
see \eqref{eq:Hacca}. We can thus summarize the behavior of $d_\alpha(t)$ in the two relevant cases as follows:

\begin{description}
\item[Non Resonant Case] We get
\begin{equation}\label{eq:dalphanr}
d_\alpha(t)=R^{\rm t}(\alpha)\sin(\Omega_+(\alpha)t)+\alpha^2r_2^{\rm t}(\alpha)\sin(\rho_+(\alpha)t)+\alpha K^{\rm
	t}_{v}(t)\ ,
\end{equation}
where
\begin{equation}\label{eq:Rt}
R^{\rm t}(\alpha)=\frac{\gamma\Omega^2
  (\Omega^2_+(\alpha)-\Omega^2)R(\alpha)}{\alpha T_P\Omega_+(\alpha)}
= \frac 1{\Omega_+(\alpha)}+O(\alpha)\ ,
\end{equation}
and a similar expression for $r^{\rm t}_2(\alpha)$. Observing that, in the notation of \eqref{eq:Bessel}, we have
\begin{equation}\label{eq:Kr}
K^{\rm t}_{v}(t)=\frac{2\alpha}{\pi}\int_{\mu_-}^{\mu_+}\frac{k^2-\Omega^2}k
\sqrt{(k-\mu_-)\left(\mu_+-k\right)}\mathcal G(k)\cos(kt)dk\ ,
\end{equation}
we get that $K^{\rm t}_{v}$ satisfies \eqref{eq:stima1} so that
\begin{equation}\label{eq:Ktv}
\alpha\int_0^\infty |K^{\rm t}_{v}(t)|dt\leq K.
\end{equation}
	
\item[Resonant Case] In this case we get
\begin{equation*}
d_\alpha(t)=|R^{\rm t}(\alpha)|\sin(\Omega_p(\alpha)t+\phi''(\alpha))e^{-\xi(\alpha)t}+\alpha^2r^{\rm
t}_2(\alpha)\sin(\rho_+(\alpha)t)+\alpha K^{\rm t}_{v}(t)\ ,
\end{equation*}
with $R^{\rm t}(\alpha)=|R^{\rm t}(\alpha)|e^{i\phi''(\alpha)}$ still given by \eqref{eq:Rt} and $K^{\rm t}_{v}$ still
satisfying \eqref{eq:Ktv}. Observe moreover that $\phi''(\alpha)=O(\alpha)$ so that we can write
\begin{equation}\label{eq:dalphar}
d_\alpha(t)=|R^{\rm t}(\alpha)|\sin(\Omega_p(\alpha)t)e^{-\xi(\alpha)t}+\alpha^2r^{\rm
t}_2(\alpha)\sin(\rho_+(\alpha)t)+\alpha \overline K^{\rm t}_{v}(t)
\end{equation}
where now $\overline K^{\rm t}_{v}(t)$ is only uniformly bounded in $t$.
	
\end{description}
	
\subsubsection{Bound on \matht{$|{\mathcal K}_N(t)-{\mathcal K}(t)|$}}

To complete the comparison between $C^{\rm t}_{\alpha,N}(s,t)$ and $C^{\rm t}_{\alpha}(s,t)$ we need to estimate the
difference between ${\mathcal K}_N(t)$ and ${\mathcal K}(t)$. Observe that ${\mathcal K}_N(0)=f_N(0)$ while 
${\mathcal K}(0)=f_+(0)$ so that, from \eqref{eq:approssimazione_f} we know that
\[
|\mathcal K_N(0)-\mathcal K(0)|=\left|\mathcal
K(0)\frac{1}{p_+(0)^{2N}-1}\right|\leq K e^{-kN}\ ,
\]
since $|p_+(0)|>1$. On the other hand we have
\[
\mathcal K(t)-\mathcal
K(0)=\frac{T_B}{2\pi}\int_0^{2\pi}\frac{\cos(\omega(\theta)t)-1}{\omega(\theta)^2}d\theta\ ,
\]
where now the integrand is an entire and periodic function of $\theta$. Again following \cite{DR1984}, we define
\begin{equation}\label{eq:fourcoeff}
\hat {\mathcal
  K}_{n}(t)=\frac{T_B}{2\pi}\int_0^{2\pi}\frac{\cos(\omega(\theta)t)-1}{\omega(\theta)^2}e^{-in\theta}d\theta\ .
\end{equation}
Observing that, for $\theta\in\mathds C$ with $|\Im(\theta)|>1$, we have
\[
 |\cos(\omega(\theta)t)|\leq \exp\left(k\tilde\omega te^{|\Im(\theta)|}\right)
\]
and shifting the integral in \eqref{eq:fourcoeff} to the segment $\theta\in[i\bar\theta,2\pi+i\bar\theta]$ with 
$\bar\theta=\ln\left(\frac {\tilde\omega t}N\right)\sgn(n)$, for $\tilde\omega t<N$ we get
\[
 \hat {\mathcal K}_{n}(t)
\leq Ke^{kN}\left(\frac{k\tilde\omega t}N\right)^{|n|}\, .
\]
Reasoning like in \eqref{eq:fourier} we get
\[
\begin{aligned}
\left|\mathcal K_N(t)-\mathcal K_N(0)-\mathcal K(t)+\mathcal K(0)\right|=&
\left|\sum_{n=-\infty}^\infty \hat{\mathcal  K}_n(t)\left( 
\frac 1{2N} \sum_{j=-N+1}^N e^{in\theta_j}-\frac1{2\pi}\int_{-\pi}^\pi e^{in\theta'}d\theta'\right)\right|=\\
&\left|\sum_{n=1}^\infty\hat{\mathcal  K}_{nN}(t)\right|\leq 
K\left(\frac{k\tilde\omega t}N\right)^{2N}\ .
\end{aligned}
\]
Since $d_\alpha(0)=0$ we finally get
\begin{equation}\label{eq:diff-t}
|C^{\rm t}_{\alpha,N}(s,t)-C^{\rm t}_{\alpha}(s,t)|\leq
K\left(\left(\frac{k\tilde\omega\max(s,t)}{N}\right)^{2N}+ t^2s^2e^{-kN}\right)
\end{equation}
Combining \eqref{eq:diff-nt} and \eqref{eq:diff-t} we obtain a complete proof of Theorem \ref{th:appo}.
\medskip

We can now turn to the study of the behavior of $C^{\rm t}_{\alpha}(s,t)$ for large $s$ and $t$. 
Using \eqref{eq:dalphanr} or \eqref{eq:dalphar} in \eqref{eq:calK} we can write $C^{\rm t}_{\alpha}(s,t)$ as a sum of 
integrals involving $\cos(\Omega_+(\alpha)t)$, $\cos(\rho_+(\alpha)t)$ and $K_v(t)$. We will show that most of 
these integrals give bounded contribution, in $\alpha$ as well as $s$ and $t$, to $C^{\rm t}_{\alpha}(s,t)$ and thus 
contribution of order $\alpha^2$ to $C_\alpha(s,t)$, see \eqref{eq:Csplit}. Indeed we first observe that
\begin{equation}\label{eq:KK}
\left|\int_0^s\int_0^t \mathcal K(t-s-\tau+\sigma)K^{\rm t}_v(\tau)K^{\rm t}_v(\sigma)d\tau d\sigma\right|\leq
K\alpha^{-2}\ ,
\end{equation}
thanks to \eqref{eq:Ktv}, so that, taking into account \eqref{eq:dalphanr} and \eqref{eq:Csplit}, we see that the above 
integral contributes a term order $\alpha^2$ to $C_\alpha(s,t)$. A more detailed analysis, sketched in Appendix 
\ref{app:stime}, shows that contribution in
\eqref{eq:KK} vanishes as a power law in $t$ and $s$. 

Observe now that for $\Xi\in\mathds R$ we have
\[
\begin{aligned}
\int_0^t \cos(\omega (t-\tau))\sin(\Xi\tau)d\tau=&\frac 
12\left(\frac{\cos((\omega+\Xi)t)-1}{\omega+\Xi}-\frac{\cos((\omega-\Xi)t)-1}{\omega-\Xi}\right)\cos(\omega t)+\\
&\frac 
12\left(\frac{\sin((\omega+\Xi)t)}{\omega+\Xi}-\frac{\sin((\omega-\Xi)t)}{\omega-\Xi}\right)\sin(\omega t)
\end{aligned}
\]
while a similar expression holds for $\int_0^t \sin(\omega(t- \tau))\sin(\Xi\tau)d\tau$. Thus, using the definition of
$\mathcal K$ in \eqref{eq:def-d-K} and \eqref{eq:Ktv}, we get
\begin{equation}\label{eq:rhorho}
\begin{aligned}
\left|\int_0^s\int_0^t \mathcal K(t-s-\tau+\sigma)\sin(\rho_+(\alpha)\tau)\sin(\rho_+(\alpha)\sigma)d\tau
d\sigma\right|&\leq K|\rho_+(\alpha)-\mu_+|^{-2}\leq K\alpha^{-4}\ ,\\
\left|\int_0^s\int_0^t \mathcal K(t-s-\tau+\sigma)K^{\rm t}_v(\tau)\sin(\rho_+(\alpha)\sigma)d\tau
d\sigma\right|&\leq K\alpha^{-1}|\rho_+(\alpha)-\mu_+|^{-1}\leq K\alpha^{-3}\, ,
\end{aligned}
\end{equation}
so that the contributions of the first and second line of \eqref{eq:rhorho} are of order $\alpha^2$ and $\alpha$ 
respectively.
Again we observe that a more detailed analysis, see Appendix
\ref{app:stime}, shows that the contribution in the second
line of \eqref{eq:rhorho} vanishes as a power law in $t$. 

Thus, as expected, the only contributions
potentially non vanishing in $\alpha$ are those containing $\Omega_+(\alpha)$.

\subsubsection{The non resonant case}\label{subsec:Ctnr}
Proceeding as for \eqref{eq:rhorho} we see that
\[
\left|\int_0^s\int_0^t \mathcal K(t-s-\tau+\sigma)\sin(\Omega_+(\alpha)\tau)\sin(\Omega_+(\alpha)\sigma)d\tau
d\sigma\right|\leq K\max_\pm|\Omega_+(\alpha)-\mu_\pm|^{-2}\leq K\ ,
\]
and similar estimates hold for the remaining terms. Summing up we get
\begin{equation}\label{eq:Ctnr}
|C^{\rm t}(s,t)|\leq K.
\end{equation}
Thus \eqref{eq:Ctnr}, together with \eqref{eq:Cntnr} and \eqref{eq:Csplit}, completes the proof of Theorem
\ref{th:outband}. Moreover, as already observed after \eqref{eq:Cntnr}, the
correction term $K(s,t)$ in \eqref{eq:Cdstima} contains a term of the form $\sin(\Omega_+(\alpha)s)\sin(\Omega_+(\alpha)t)$.
We can thus conclude that
\begin{equation}\label{eq:infsupnr}
\limsup_{\tau\to\infty}C_\alpha(\tau,\tau+t)- \liminf_{\tau\to\infty}C_\alpha(\tau,\tau+t)=O(\alpha)\, .
\end{equation}

\subsubsection{The resonant case}

Since in this case $\Omega(\alpha)$ is close to the real segment $[\mu_-,\mu_+]$  we have to be more careful. To this 
extent we write 
\eqref{eq:expCt} as
\begin{equation}\label{eq:calC}
	C^{\rm
		t}_{\alpha}(s,t)=\frac{T_B}\pi\int_{\mu_-}^{\mu_+}\frac{d\omega}{\omega^2}\frac{d\theta}{d\omega}\int_0^s
	\int_0^t \cos(\omega(t-s-\tau+\sigma))d_{\alpha}(\tau)d_{\alpha}(\sigma)d\tau d\sigma
\end{equation}
where
\begin{equation}\label{eq:dtdo}
\frac{d\theta}{d\omega}=\frac{2\omega}{\sqrt{(\omega^2-\mu_-^2)(\mu_+^2-\omega^2)}}\, .
\end{equation}
Expanding $d_\alpha(t)$ using \eqref{eq:dalphar} the contribution containing 
$\sin(\Omega(\alpha)t)\sin(\Omega(\alpha)s)$ is the most relevant. Using Lemma \ref{lem:duhamel} we get
\begin{equation}\label{eq:Ctr}
\begin{aligned}
\int_{\mu_-}^{\mu_+}&\frac{d\omega}{\omega^2}\frac{d\theta}{d\omega}\int_0^s \int_0^t
\cos(\omega(t-s-\tau+\sigma))e^{-\xi(\alpha)\tau}
\sin(\Omega_p(\alpha)\tau)e^{-\xi(\alpha)\sigma}\sin(\Omega_p(\alpha)\sigma)=\\
&\frac{\pi}{2\xi(\alpha)}\frac{1}{\Omega_p(\alpha)\sqrt{(\mu_+-\Omega_p(\alpha))(\mu_--\Omega_p(\alpha))}}  
\left(e^{-\xi(\alpha)|t-s|}-e^{-\xi(\alpha)(t+s)}\right)\cos(\Omega_p(\alpha)(t-s))+
K(t,s)\,.
\end{aligned}
\end{equation}
Finally, following the scheme of the proof of Lemma \ref{lem:duhamel}, it is easy to see the the remaining contributions
to $C_\alpha^{\rm nt}$ are uniformly bounded in $\alpha$, $s$ and $t$. Thus \eqref{eq:Ctr}, together with
\eqref{eq:Cntr} and \eqref{eq:Csplit}, completes the proof of Theorem \ref{th:inband}.

From Appendix \ref{app:stime} we see that for $t$ and $s$ going to infinity, the only two contributions to $C_\alpha^{\rm
t}(s,t)$ that do not vanish are oscillations of the form
$\sin(\rho_+(\alpha)(t+s))$ or $\sin(\rho_+(\alpha)(t-s))$, together with
the term $e^{-\xi(\alpha)|t-s|}\cos(\Omega_p(\alpha)(t-s))$, if $t-s$
remain finite. We thus get
\[
\limsup_{\tau\to\infty}C^{\rm t}_\alpha(\tau,\tau+t)-
\liminf_{\tau\to\infty}C^{\rm t}_\alpha(\tau,\tau+t)=O(\alpha^3)\ ,
\]
so that, considering \eqref{eq:Csplit} and \eqref{eq:CtMark}, we obtain
\begin{equation}\label{eq:nonth}
\limsup_{\tau\to\infty}C_\alpha(\tau,\tau+t)- \liminf_{\tau\to\infty}C_\alpha(\tau,\tau+t)=O(\alpha^5)\, .
\end{equation}

\section{Discussion and outlook}
\label{sec:discussion}

\red {From Remark \ref{rmk:1timeres} we get a detailed description of the one time correlation function $C_{\alpha}(0,t)$
in the resonant case that may be of interest if one wishes to check numerically our predictions. If we assume that
$\alpha$ is small, so that the three terms in \eqref{eq:C1r} are clearly distinguishable, we see that
$C_{\alpha}(0,t)$ initially decays exponentially for a time of the order of $\alpha^{-2}$. It then start decaying
slowly until it settles on to a very small oscillation of the order of $\alpha^3$. On the other hand, Theorem
\ref{th:appo} guarantees that $C_{\alpha,N}(0,t)$ stays close to $C_{\alpha}(0,t)$ only for time of order $N$.
Since the Hamiltonian \eqref{eq:Ham0} is harmonic, we know that $C_{N,\alpha}(0,t)$ is a quasi--periodic function
of $t$. Although in Theorem \ref{th:appo} we do not have a lower bound, we expect that
$|C_{\alpha,N}(0,t)-C_{\alpha}(0,t)|$ grows rapidly to be order 1 for times larger than $N$. Thus for $N$ of order
of $\alpha^{-2}$, $C_{\alpha,N}(0,t)$ will follow the initial exponential decay of $C_{\alpha}(0,t)$ before
departing. Clearly the larger $N$ the larger the portion of exponential decay one can observe. On the other hand, a
much larger $N$ is needed to observe the slowly decaying corrections or the steady oscillation in \eqref{eq:C1r}.
An analogous but more complex control on the behavior of $C_{\alpha}(s,t)$ follows from Appendix \ref{app:stime}.}

\medskip

It is natural to wonder how much our results depend on the specific form of the model we have decided to consider, that
is, on the form of the Hamiltonian \eqref{eq:Ham0}. If we insist on the full dynamics to be linear, there is little
freedom for $H_P$. Regarding $H_B$ we can consider a more general translation invariant potential by taking
\[
 H_B(\hat q,\hat p)=\sum_{l=-N+1}^{N} \frac{\hat
  p_l^2}{2m}+\sum_{l,m=-N+1}^{N} \hat V_{|l-m|} \hat q_l\hat q_m\ ,
\]
where, for simplicity sake, we assume that $V$ has finite range $L$, i.e., $V_k=0$ for $k>L$. In
this case, the normal modes of the Hamiltonian are still given by \eqref{eq:normal} while the
frequencies satisfy $\omega^2_j=\omega^2(\pi j/N)$ with $\omega^2(\theta)$ a trigonometric polynomial, see 
\eqref{eq:frequenze_catena}.

We can now repeat our analysis up to Subsection \ref{subsec:exact} and, in this more general case we get
\begin{equation}\label{eq:f+gen}
 f_+(\lambda)=\sum_{\genfrac{}{}{0pt}{2}{k:\Im(\theta_k)>0}{ 
 \omega^2(\theta_k)=-\lambda^2}}\frac{i}{\tfrac{d}{d\theta}\omega^2(\theta_k)}\ ,
\end{equation}
so that the function $f_+(i\omega)$ is strictly linked with the density of states around the frequency $\omega$.
Equation \eqref{eq:f+gen} also makes it clear that the analytic structure of $f_+$ near $\mathcal I$, and thus the
geometry of the Riemann surface $\mathcal F$, depends on the number of solutions $\theta_k\in\mathds R$ of
$\omega^2(\theta_k)=-\lambda^2$ for $\lambda\in\mathcal I$. If we assume that $\omega^2(\theta)$ is strictly increasing
in $(0,\pi)$, for example requiring $V_k$ to be small for $k>1$, then the analysis in Sections \ref{sec:decay1} and
\ref{sec:decay2} can proceed without modifications. In the general case, new branch points may appear in $\mathcal
F$ in coincidence with the maxima and minima of $\omega^2(\theta)$. This will not qualitatively change the analysis in 
subsection \ref{subsec:r} and thus the behavior of the correlation in the non resonant case. On the contrary, a more 
detailed analysis is needed if $i\Omega$ is close to one such branch points but this is outside the scope of this 
paper.

\medskip

\red{As another possible extension we can consider a different initial distribution for the initial condition of the bath. Indeed, \eqref{eq:dens_prob} corresponds to equipartition in the bath but there 
are many other invariant distributions for the evolution with $\alpha=0$. More generally, we can consider the initial 
distribution
\begin{equation}\label{eq:nonther}
\tilde\rho_{N}(q,p,Q,P)=\frac{1}{Z(T_B,T_P)} \exp\left(-\frac12\sum_{j=0}^N\frac 1{T_B(j)}
(p_j^2+\omega_j^2 q_j^2)-\frac 1{2T_P} (P^2+\Omega^2Q^2)\right)\ .
\end{equation}
By properly choosing $T_B(j)$, the time correlation generated by \eqref{eq:nonther} are analogous to those considered 
in \cite{Das} for a quantum bath.}

\red{The only relevant change to our computations appears in Subsection \ref{subsec:C22} 
where 
${\mathcal K}_N(t)$ must be replaced by
\[
\widetilde{\mathcal K}_N(t)=\frac{1}{N}\sum_{j=0}^N 
\frac{T_B(j)\eta_j^2\cos(\omega_jt)}{\omega_j^2}
\]
Assuming that we can write $T_B(j)=T(\omega_j)$ with $T$ smooth we get
\[
\widetilde{\mathcal K}_N(t)=\frac{1}{2N}\sum_{j=-N+1}^N
\frac{T(\omega_j)\cos(\omega_jt)}{\omega_j^2}\qquad\hbox{and}\qquad
\widetilde{\mathcal K}(t)=\frac{1}{2\pi}\int_0^{2\pi}
\frac{T(\omega(\theta))\cos(\omega(\theta)t)}{\omega(\theta)^2}d\theta\, .
\]
We can now repeat the analysis in Section \ref{sec:decay2} with only minor changes. Observe though that in the resonant 
case the probe will thermalize at the temperature of the oscillator in the bath with frequency $\Omega(\alpha)$, that 
is, in \eqref{eq:Cdstimain} we have $T(\Omega(\alpha))$ instead of $T_B$.
}

\medskip

At last, a natural extension would be to investigate how the results change when a macroscopic probe is considered, that
is, a system with $M$ degrees of freedom, with $1\ll M\ll N$, and the dependence of the estimates on $M$, $N$, and the
relative size $M/N$. The case is of particular interest when the probe is composed by some degrees of freedom resonating
with the bath, and some not at all. We expect the probe then to split somehow into two subsystems, one thermalizing with
the bath, and the other one preserving the initial temperature, leading to an occurrence of incomplete thermalization,
similar to that of diatomic gases (see \cite{BCG95}). This would not be surprising in the fully linear case, while if a
nonlinear perturbation is introduced, a similar behavior has been numerically observed for a linear chain in contact
with a perfect gas thermostat in \cite{CCG04}. For an analytical treatment, we plan to consider a diatomic chain for the
probe, where the optical and acoustical branch have well separated frequencies, and we have already at hand results
guaranteeing that the internal dynamic of the macroscopic probe do not allow energy exchanges between the branches, even
in the thermodynamic limit (see \cite{Mai19}).
\bigskip

\noindent
{\bf Acknowledgements}. The authors wish to thank Andrea Carati, Livia Corsi e Luigi Galgani for their help and support. F.B.~has been partially supported by the NSF grant DMS-1907643.

\medskip

\noindent
{\bf Data Availibility Statement}.
Data sharing is not applicable to this article as no datasets were generated or analyzed
during the current study.

\medskip

\noindent
{\bf Conflict of Interest}.
The authors have no conflict of interest to declare.

\bibliographystyle{unsrt}


\appendix

\section{Technical Lemmas}
\label{app:lemmas}

In this appendix we collect few results that are used many times in the course of the paper.

As discussed in Remark \ref{rmk:shift}, we want to compute \eqref{eq:antiinf} by shifting the integral from a line with
real part $\xi>0$ to a line with real part $\xi<0$. To do this, we need to take into account the discontinuity of
$g^1_\alpha(\lambda)$ for $\lambda\in \mathcal I$. This is the purpose of the following Lemma.

\begin{lemma}\label{lem:stima} 
Let $h(\xi)$ be analytic in $\mathcal R=\{\xi; |\Re(\xi)|< 1\}$ and for $d>0$, let $\mathcal R_d=\{\xi\in \mathcal R;
|\Im(\xi)|\leq d\}$. Then, for $\e>-1$ and $t\ge 0$ we have
\begin{equation}\label{eq:fint}
\left |\int_{-1}^1 (1-\xi^2)^{\e}h(\xi)e^{i \xi t}d\xi\right|\leq
\frac{K_{d,\e}\sup_{\xi\in \mathcal R_d}|h(\xi)|}{\left(1+dt\right)^{1+\e}}\ .
\end{equation}

\end{lemma}

\begin{proof}
Consider the path $\Xi_d(\chi)=\chi+ id(1-|\chi|)$, with
$\chi\in[-1,1]$. We get, for $t>0$, %
\begin{equation}\label{eq:stima}
\begin{aligned}
\left |\int_{-1}^1 (1-\xi^2)^{\e}h(\xi)e^{i \xi t}d\xi\right|
=&\left|\int_{\Xi_d}(1-\xi^2)^\e h(\xi) e^{i\xi t}d\xi\right|\leq\\
&2(1+d^2)^{(1+\e)/2}\sup_{\xi\in \Xi_d}\left|h(\xi)\right|\int_{0}^{1}
(1-\chi)^{\e}e^{-d(1-\chi)t}d\chi\leq\\
&2\left(\frac{\sqrt{1+d^2}}{d}\right)^{1+\e}\sup_{\xi\in \Xi_d}\left|h(\xi)\right|
{t^{-(1+\epsilon)}}\int_0^{td}s^{\epsilon}e^{-s}ds\ .
\end{aligned}
\end{equation}
This, together with the trivial case $t=0$, completes the proof with
\[
K_{d,\e}=K\left(1+d^2\right)^{\frac{1+\e}2}\Gamma(1+\e)\ ,
\]
with $\Gamma$ denoting the gamma function.
\end{proof}

We list here a couple of easy consequences. As discussed at the end of Subsection \ref{subsec:galpha}, in the resonant
case, the poles at $\Omega_\pm(\alpha)$ are close to the set $\mathcal I$. In that situation we will use the following
Corollary.

\begin{coro}\label{cor:stimasing}
 Let $h(\xi)$ be analytic in $\mathcal R$. Then for $\zeta\in \mathcal R$ with $0<\Im(\zeta)\leq 1$ and $\e>-1$ we have
\begin{equation}\label{eq:fintsing}
\left |\int_{-1}^1 \frac{(1-\xi^2)^{\e}h(\xi)}{\xi-\zeta}e^{i
\xi t}d\xi-2\pi i(1-\zeta^2)^\e h(\zeta)e^{i\zeta t}\right|\leq
\frac{K'_{\e}\sup_{\xi\in
\mathcal R_{2/(1-|\Re(\zeta)|)}}|h(\xi)|}{(1-|\Re{\zeta}|)^{2+\epsilon}(1+t)^{1+\e}}\ .
\end{equation}
\end{coro}

\begin{proof}
Take $d=2(1-|\Re(\zeta)|)^{-1}>1$ so that $\zeta$ is in the domain bounded by the segment $[-1,1]$ and the path $\Xi_d$. Moreover we have 
\[
\inf_{\xi\in \Xi_d} |\xi-\zeta|=\frac{d(1-|\Re(\zeta)|)-\Im(\zeta)}{\sqrt{1+d^2}}\geq 
\frac{1}{\sqrt{1+d^2}}
\]
so that the thesis follows with
\[
K'_{\e}=K\Gamma(1+\e)\ .
\]
\end{proof}

We turn now to a Lemma that will be used to evaluate the two time correlation function in both the Hamiltonian and
stochastic case.

\begin{lemma}\label{lem:duhamel}
Let $0<a<b$, let $\dam$, $\Xi$ be such that $\Xi\in(a,b)$, while $0<\dam\leq \min\{\Xi-a,b-\Xi\}/2$ and let $g(\omega)$ be a function real for
$\omega\in [a,b]$ and analytic in $\mathcal R_{a,b}= \{\xi: a\le
\Re(\xi)\le b, |\Im(\xi)|\le \max\{1,\dam\}\}$. Then we have,
\begin{equation}\label{eq:duha}
\begin{aligned}
 \int_{a}^{b}d\omega\, &\frac{g(\omega)}{\sqrt{(\omega-a)(b-\omega)}} 
 \int_0^t\int_0^s\cos(\omega(t-s-\tau+\sigma))
 e^{-\dam \tau}\sin(\Xi\tau)e^{-\dam\sigma}\sin(\Xi\sigma)d\sigma d\tau=\\
&\frac{\pi}{4\dam}\frac{g(\Xi)}{\sqrt{(\Xi-a)(b-\Xi)}}  
\left(e^{-\dam|t-s|}-e^{-\dam(t+s)}\right)\cos(\Xi(t-s))+
\frac{K(t,s)}{\left((\Xi-a)(b-\Xi)\right)^{2+\e}}\ ,
\end{aligned}
\end{equation}
where $K(t,s)$ is bounded uniformly in $\beta$, $s$ and $t$.
\end{lemma}

\begin{proof} Observe that
\begin{equation}\label{eq:contr1}
\begin{aligned}
\int_0^td\tau\int_0^sd\sigma&
\cos(\omega(t-s-\tau+\sigma))
\sin(\Xi\tau)e^{-\dam \tau}\sin(\Xi\sigma)e^{-\dam \sigma}=\\
&-\frac 18\sum_{g_1,g_2,g_3=\pm}g_2g_3\int_0^td\tau\int_0^sd\sigma
g_2g_3e^{g_1i\omega(t-s-\tau+\sigma)}
e^{(g_2i\Xi-\dam)\tau}
e^{(g_3i\Xi-\dam)\sigma}=\\
&-\frac 18\sum_{g_1,g_2,g_3=\pm}g_2g_3e^{g_1i\omega(t-s)}
\int_0^t e^{[i(-g_1\omega+g_2\Xi)-\dam]\tau}d\tau
\int_0^s e^{[i(g_1\omega+g_3\Xi)-\dam]\sigma}d\sigma=\\
&-\frac 18\sum_{g_1,g_2,g_3=\pm}g_2g_3e^{g_1i\omega(t-s)}
\frac{1-e^{[i(-g_1\omega+g_2\Xi)-\dam]t}}
{i(-g_1\omega+g_2\Xi)-\dam}\;
\frac{1-e^{[i(g_1\omega+g_3\Xi)-\dam]s}}
{i(g_1\omega+g_3\Xi)-\dam}=\\
&-\frac 18\sum_{g_1,g_2,g_3=\pm}g_2g_3
\frac{e^{ig_1\omega t}-e^{[ig_2\Xi-\dam]t}}
{\omega-g_1g_2\Xi-ig_1\dam}\;
\frac{e^{-ig_1\omega s}-e^{[ig_3\Xi-\dam]s}}
{\omega+g_1g_3\Xi+ig_1\dam}\ .
\end{aligned}
\end{equation}
Notice that the terms for  $g_2=-g_1$ and $g_3=g_1$ are bounded
uniformly in $\omega\in[a,b]$ and $\dam$,$s$ and $t$. If $g_2=g_3$ the
corresponding term have one pole
close to $[-1,1]$ but the residue is 
bounded
by a constant independent of $\dam$. Using Corollary \ref{cor:stimasing} we obtain that also contribution of this
term to \eqref{eq:duha} can be bounded uniformly in $\dam$, $s$ and $t$. 

We are thus left with the contribution for $g_2=g_1=-g_3$, 
\begin{equation}\label{eq:contazzo}
\begin{aligned}
&\frac 18\sum_{g_1=\pm}\int_{a}^{b}d\omega\, \frac{g(\omega)}{\sqrt{(\omega-a)(b-\omega)}}
\frac{\left(e^{ig_1\omega t}-e^{[ig_1\Xi-\dam]t}\right)\left(e^{-ig_1\omega 
s}-e^{[-ig_1\Xi-\dam]s}\right)}{(\omega-\Xi)^2+\beta^2}=\\
&\frac 18\int_{a}^{b}d\omega\, \frac{g(\omega)}{((\omega-\Xi)^2+\dam^2) \sqrt{(\omega-a)(b-\omega)}}\\
&\qquad\qquad\qquad\qquad\sum_{g_1=\pm}\left(e^{ig_1\omega(t-s)}-e^{(ig_1\Xi-\dam)t-ig_1\omega s}-e^{ig_1\omega t - 
(ig_1\Xi+\dam)s}+
e^{ig_1\Xi(t-s)-\dam(t+s)}\right)=\\
&\frac {\pi}{4\dam}\frac{g(\Xi)}{\sqrt{(\Xi-a)(b-\Xi)}}e^{-\dam|t-s|}\cos(\Xi(t-s))+\\
&\qquad\qquad
-\frac
\pi{4\dam}\frac{g(\Xi)}{\sqrt{(\Xi-a)(b-\Xi)}}
e^{-\dam(t+s)}\cos(\Xi(t-s))+\frac{K(t,s)}{\left((\Xi-a)(b-\Xi)\right)^{2+\e}}\ , 
\end{aligned}
\end{equation}
where $K(t,s)$ is uniformly bounded in $\dam$, $t$, and $s$ and vanishes when $t$ or $s$ go to infinity. In 
\eqref{eq:contazzo} we have expanded the product to be able to apply Lemma \ref{lem:stima} and Corollary 
\ref{cor:stimasing}. Indeed, in each term, whether to move the $\omega$ integration path for $\Im(\omega)$ positive or 
negative depends on the sign of the factor multiplying $i\omega$ in
the exponent.

\end{proof}

From the proof of Lemma \ref{lem:duhamel} we immediately get.

\begin{coro}\label{cor:duhamel}
Let $0<a<b$, let $\dam$, $\Xi$ be such that $a<\Xi<b$ and $0<\dam\leq \min\{\Xi-a,b-\Xi\}/2$ and let $g(\omega)$ be a function real
for $\omega\in [a,b]$ and analytic in $\mathcal R_{a,b}= \{\xi: a\le\Re(\xi)\le b, |\Im(\xi)|\le
\max\{1,\dam\}\}$. Then we have, 
\begin{equation}\label{eq:duhas}
\begin{aligned}
 \int_{a}^{b}d\omega\,& \frac{g(\omega)}{\sqrt{(\omega-a)(b-\omega)}} 
 \int_0^t\int_0^s\cos(\omega(t-s-\tau+\sigma))
 e^{-\dam \tau}\cos(\Xi\tau)e^{-\dam\sigma}\cos(\Xi\sigma)d\sigma d\tau=\\
&\frac{\pi}{4\dam}\frac{g(\Xi)}{\sqrt{(\Xi-a)(b-\Xi)}}  
\left(e^{-\dam|t-s|}-e^{-\dam(t+s)}\right)\cos(\Xi(t-s))+
\frac{K(t,s)}{\left((\Xi-a)(b-\Xi)\right)^{2+\e}}\ ,
\end{aligned}
\end{equation}
where $K(t,s)$ is bounded uniformly in $\beta$, $s$ and $t$.
\end{coro}

\section{Better estimates for the long time behavior of \matht{$C^{\rm t}_\alpha(s,t)$}}\label{app:stime}

In this appendix we extend the analysis in Subsection \ref{subsec:asym} to obtain better estimates
for the long time behavior of $C^{\rm t}_\alpha(s,t)$. We will mostly use it to discuss the limit
of $C_\alpha^{\rm t}(\tau,\tau+t)$ when $\tau\to\infty$.  As a first step we rewrite
\eqref{eq:expCt}. Defining
\[
\begin{aligned}
	\matC(t,\omega):=&\int_0^t
	d_{\alpha}(\tau)\cos\omega(t-\tau)d\tau=\frac{1}{2\pi
		i}\int_{\xi-i\infty}^{\xi+i\infty}\frac{\lambda e^{\lambda
			t}
          d\lambda}{(\lambda^2+\omega^2)D_{\alpha}(\lambda)}\ ,\\
	\matS(t,\omega):=&\int_0^t
	d_{\alpha}(\tau)\sin\omega(t-\tau)d\tau=\frac{1}{2\pi
		i}\int_{\xi-i\infty}^{\xi+i\infty}\frac{\omega e^{\lambda
			t}
          d\lambda}{(\lambda^2+\omega^2)D_{\alpha}(\lambda)}\ ,
\end{aligned}
\]
after some algebra, we can rewrite \eqref{eq:expCt} as
\begin{equation}\label{eq:CS}
C^{\rm t}_\alpha(s,t)=\frac{T_B}{2\pi}\int_{\mu_-}^{\mu_+}\frac{d\omega}{\omega^2}\frac{d\theta}{d\omega}
\bigl(\matC(s,\omega)\matC(t,\omega)+
\matS(s,\omega)\matS(t,\omega)\bigr)\ .
\end{equation}
We thus need to understand the behavior of  $\matC(t,\omega)$ and $\matS(t,\omega)$ for $t$ large and $\omega$ close 
to $\mu_\pm$.                          

Reasoning as in Subsection \ref{subsec:asym} we get
$$
\matC(t,\omega)=
\Re{\left[\matC_1(t,\omega)+\matC_2(t,\omega)+
\matC_3(t,\omega)\right]}+\matC_4(t,\omega) \ ,
$$
with
\begin{equation}\label{eq:matC}
  \begin{split}
\matC_1(t,\omega)&:=\frac{\Omega_+(\alpha)R^{\rm 
    t}(\alpha)}{\omega^2-\Omega_+^2(\alpha)}e^{i\Omega_+(\alpha)t}\ ,
  \quad \matC_2(t,\omega):=\frac{\alpha^2
    \rho_+(\alpha)r^{\rm 
      t}_2(\alpha)}{\omega^2-\rho_+^2(\alpha)}e^{i\rho_+(\alpha)t}\ ,\\\quad
\matC_3(t,\omega)&:=\frac{e^{i\omega t}}{D_\alpha(i\omega+0^+)}\ ,
  \end{split}
\end{equation}
where $R^{\rm t}(\alpha)$ and $r^{\rm t}_2(\alpha)$ are discussed in
Subsection \ref{subsec:dalpha}, while
\begin{equation}\label{eq:def_matK}
\begin{aligned}
\matC_4(t,\omega)=&\int_{\Xi^+_d\cup\Xi^-_d}\frac{\xi}{\xi^2-\omega^2}
\frac{\alpha\gamma(\xi^2-\Omega^2)\sqrt{(\mu_+^2-\xi^2)(\x^2-\mu_-^2)}}
{(\xi^2-\bar\Omega^2)^2(\mu_+^2-\xi^2)(\xi^2-\mu_-^2)+\alpha^2\gamma^2(\xi^2-\bar\Omega^2)}e^{i\omega \xi}d\x=\\
&\alpha\int_{\Xi^+_d\cup\Xi^-_d}\frac{\sqrt{(\mu_+-\xi)(\x-\mu_-)}}
{(\xi-\omega)(\xi-\rho_-(\alpha))(\rho_+(\alpha)-\xi)}\matG(\xi)e^{i\omega \xi}d\xi
\end{aligned}
\end{equation}
where $\Xi^\pm_d=\{\pm\bar\mu+\delta_\mu \chi+id(1-|\chi|)\, ;\,
\chi\in[-1,1]\}$, see \eqref{eq:Bessel}, and $\matG$ is analytic for
$\xi$ near the integration path.

In an analogous way we can write
    $$
    \matS(t,\omega)=\Im{\left[\matS_1(t,\omega)+\matS_2(t,\omega)+
\matS_3(t,\omega)\right]}+\matS_4(t,\omega) \ ,
    $$
    by defining
\begin{equation*}
  \begin{split}
    \matS_1(t,\omega)&:=\frac{\omega R^{\rm 
        t}(\alpha)}{\omega^2-\Omega_+^2(\alpha)}e^{i\Omega_+(\alpha)t}\ ,\quad
    \matS_2(t,\omega):=\frac{\alpha^2
  \omega r^{\rm 
    t}_2(\alpha)}{\omega^2-\rho_+^2(\alpha)}e^{i\rho_+(\alpha)t}\\
    \matS_3(t,\omega)&:=\frac{e^{i\omega t}}{D_\alpha(i\omega+0^+)}=\matC_3(t,\omega)\ ,
  \end{split}
\end{equation*}
and $\matS_4(t,\omega)$ is similar to $\matC_4(t,\omega)$, the
only difference being a factor $-i\omega$ in place of $\xi$ in the
first line of \eqref{eq:def_matK}.

We notice that from Lemma \ref{lem:stima} we get
\begin{equation}\label{eq:matK}
  \begin{split}
\matC_4(t,\omega)&\leq\sup_{\xi\in\Xi_d}\left|\frac{\sqrt{(\mu_+-\xi)(\x-\mu_-)}}
{(\xi-\omega)(\xi-\rho_-(\alpha))(\rho_+(\alpha)-\xi)}\right|\frac{\alpha
  K}{\sqrt{t}}\\
&\leq \frac {1} 
{\sqrt{((\omega-\mu_+)^2+\alpha^2)((\omega-\mu_-)^2+\alpha^2)}}\frac{\alpha
  K}{\sqrt{t}}\ ,
\end{split}
\end{equation}
where we have
used \eqref{eq:poletto}, and an analog inequality for $\matS_4$.  Observe that, for $t$ large, we can improve the above estimate along the lines of
\eqref{eq:optimal}, but this will worsen the behavior in $\omega$ near $\mu_\pm$, giving a potentially diverging $\omega$ integral  
in \eqref{eq:CS}. We will thus not try to optimize the $t$ behavior in \eqref{eq:matK}.

Concerning the uniformity in $\alpha$ of the estimates, by \eqref{eq:matK} and
\eqref{eq:poletto}, we have that
\begin{equation}\label{eq:app_unif_alfa}
\matC_4(t,\omega)\le K/\sqrt{t}\ ,\quad \matS_4(t,\omega)\le K/\sqrt{t}
\ ,\quad  \tfrac{\alpha^2}{\omega^2-\rho_+^2(\alpha)}\le K\ ,\qquad
\forall \omega\in[\mu_-,\mu_+]\ .
\end{equation}

\subsection{The non resonant case}

To study the long time behavior of the correlations through \eqref{eq:CS}, we must consider the $\omega$ integrals of
all possible product of a $\matC_i(\omega,t)$, or its complex conjugate, with a $\matC_j(\omega,s)$, together with all
similar products involving $\matS_i(\omega,s)$ and $\matS_j(\omega,s)$. Since to every product involving the $\matS_i$
can be associated with a analogous product involving the $\matC_i$, we will study only the possible pairings involving
the $\matC_i$. We immediately notice that, by virtue of \eqref{eq:matK}, all the terms containing $\matC_4$ multiplied
by anything else give a vanishing contribution to $\lim_{\tau\to \infty}C_\alpha^{\rm t}(\tau,\tau+t)$. The same applies
to all terms coming from a product where $\matC_3$ appears at least once, as it follows easily by Lemma~\ref{lem:stima};
the only exceptions are the terms (coming from a pairing of $\matC_3$ and its complex conjugate) which depend on
$\omega(t-s)$ only and decay in $t-s$ as a power law. The terms coming from the product of $\matC_1$ or $\matC_2$
multiplied by $\matC_1$ or $\matC_2$ or their complex conjugate give rise to an oscillating sinusoidal term, with a
frequency which is a combination of $\Omega_+(\alpha)$ and $\rho_+(\alpha)$. Notice that, due to
\eqref{eq:app_unif_alfa} and the fact that $\omega^2-\Omega_+^2(\alpha)$ is bounded from below uniformly in $\alpha$,
all terms are uniformly bounded in $\alpha$.

\subsection{The resonant case} In this case there are no substantial changes concerning the terms involving $\matC_2$
(or $\matS_2$) and $\matC_4$ (or $\matS_4$). The main difference comes the terms involving $\matC_1$ (or $\matS_1$) and
$\matC_3=\matS_3$ that present a singularity for $\omega=\Omega_+(\alpha)$ which is close to the integration path. Again
we will consider only the terms coming from the $\matC_i$, leaving the analog terms for the $\matS_i$ to the reader.

The integral over $\omega$ of the product of $\matC_4$ with anything else is bounded by
$\alpha^{-1}/\sqrt{\min{(t,s)}}$, due to \eqref{eq:matK}, \eqref{eq:app_unif_alfa} and the fact that
\begin{equation*}
  \begin{split}
\sup_{\omega_\in[\mu_-,\mu_+]}\frac{1}{\omega^2-\Omega_+^2(\alpha)}
\frac1{\sqrt{((\omega-\mu_+)^2+\alpha^2)((\omega-\mu_-)^2+\alpha^2)}}
\le K\alpha^{-2}\ ,\\
\sup_{\omega_\in[\mu_-,\mu_+]}\frac{1}{D_\alpha(i\omega+0^+)}
\frac1{\sqrt{((\omega-\mu_+)^2+\alpha^2)((\omega-\mu_-)^2+\alpha^2)}}
\le K\alpha^{-2}\ .
  \end{split}
  \end{equation*}

Whenever we consider $\matC_1$ multiplied by a term $\matC_1$, $\matC_2$, $\matC_3$ or their complex conjugates we can
move the integration path away from the singularity in $\Omega_+(\alpha)$, so that all term but two are bounded in
$\alpha$ and decay exponentially as $Ke^{-\xi(\alpha) \min{(t,s)}}$. The two exceptions are:
\begin{enumerate}
\item 	the product of $\matC_1$ with its complex conjugate where the denominator is
$(\omega^2-\Omega_+^2(\alpha))(\omega^2-\left(\Omega_+^*(\alpha)\right)^2)$: this gives rise to an oscillating term in
$(t-s)$, exponentially decreasing as $K\alpha^{-2}e^{-\xi(\alpha)(t+s)}$.

\item the product of $\matC_1$ times $\matC_3^*$ (or $\matC_1^*$ times $\matC_3$), where the denominators
$(\omega^2-\Omega_+^2(\alpha))D_\alpha(-i\omega+0^+)$ or
$(\omega^2-\left(\Omega_+^*(\alpha)\right)^2)D_\alpha(i\omega+0^+)$ appear: here we get a term bounded by
$K\alpha^{-2}e^{-\xi(\alpha)\min{(t,s)}}$.

\end{enumerate}

Concerning the product of $\matC_3$ with $\matC_2$ or its complex conjugate, we can always apply Lemma~\ref{lem:stima}
to bound the contribution as $K/\sqrt{\min{(t,s)}}$. The same applies to the product of $\matC_3$ with $\matC_3$ itself,
depending on $e^{i\omega(t+s)}$ only, while in the product of $\matC_3$ by its complex conjugate we encounter a term
depending on $e^{i\omega(t-s)}$, where we cannot move the integration path. Here we get an oscillating term, bounded by
$K\alpha^{-2} e^{-\xi(\alpha)|t-s|}$, which is exactly the dominating term in the long run, for fixed $|t-s|$ (cfr.
\eqref{eq:Ctr}).

The last remaining terms are the integrals of the products of $\matC_2(\omega,t)$ with $\matC_2(\omega,s)$. Here we get
contributions that do not decay with $t$ or $s$ and we need a finer bound to isolate the contribution depending on
$t-s$. For this reason we must consider in full the sum of $\matC_2$ and $\matS_2$ terms. We have
\begin{equation*}
  \begin{split}
    \Re\matC_2(t,\omega)\Re\matC_2(s,\omega)+\Im\matS_2(t,\omega)\Im\matS_2(s,\omega)&\\
    =\left(\frac{\alpha^2 \rho_+(\alpha)r^{\rm 
t}_2(\alpha)}{\omega^2-\rho_+^2(\alpha)}\right)^2\cos(\rho_+(\alpha)(t-s))+&
\frac{(\alpha^2r^{\rm 
t}_2(\alpha))^2}{\omega^2-\rho_+^2(\alpha)}\sin(\rho_+(\alpha)t)
\sin(\rho_+(\alpha)s) \ .
  \end{split}
\end{equation*}
In integrating such functions over $\omega$, we use the fact that, by virtue of
\eqref{eq:poletto}, 
$$
\int_{\mu_-}^{\mu_+}
\frac{d\theta}{d\omega}\frac{1}{(\omega^2-\rho_+^2(\alpha))^2} d\omega
= O(\alpha^{-3})\ ,\quad \int_{\mu_-}^{\mu_+}
\frac{d\theta}{d\omega}\frac{1}{\omega^2-\rho_+^2(\alpha)} d\omega
= O(\alpha^{-1})\ .
$$
This entails that in $C^{\rm t}_\alpha(s,t)$ there is an
oscillating term  with frequency $\rho_+(\alpha)$, depending on $t-s$ only, of size $\alpha$, plus an oscillating term with the same frequency, depending on $t+s$ only, of
size $\alpha^3$. 


\end{document}